\documentclass[12pt]{article}
\usepackage{amsmath,amssymb}
\usepackage{amsthm}

\def\LC{\mathcal{L}}

\def\HC{\mathcal{H}}

\def\C{\mathbf{C}}
\def\E{\mathbf{E}}

\def\N{\mathbf{N}}
\def\P{\mathbf{P}}
\def\R{\mathbf{R}}
\def\Z{\mathbf{Z}}

\def\1{\mathbf{1}}

\def\tr{\rm{tr}}

\def\al{\alpha}
\def\be{\beta}
\def\pa{\partial}
\def\ep{\epsilon}
\def\de{\delta}
\def\ga{\gamma}
\def\ka{\varkappa}

\newtheorem{prop}{Proposition}[section]
\newtheorem{theorem}{Theorem}[section]
\newtheorem{lemma}{Lemma}[section]

\newtheorem{remark}{Remark}

\newcommand{\la}{\lambda}
\newcommand{\si}{\sigma}

\newcommand{\om}{\omega}
\newcommand{\Ga}{\Gamma}
\newcommand{\De}{\Delta}

\newcommand{\Om}{\Omega}

\begin{document}
\title{Continuous time random walks modeling of quantum measurement
and fractional equations of quantum stochastic filtering and control}

\author{Vassili N. Kolokoltsov\thanks{Department of Statistics, University of Warwick,
 Coventry CV4 7AL UK, associate member of HSE, Moscow,
  Email: v.kolokoltsov@warwick.ac.uk}}
\maketitle
	
\begin{abstract}
Initially developed in the framework of quantum stochastic calculus, the main equations of quantum
 stochastic filtering were later on derived as the limits of Markov models of discrete measurements
 under appropriate scaling. In many branches of modern physics it became popular to extend random
 walk modeling to the continuous time random walk (CTRW) modeling, where the time between discrete
events is taken to be non-exponential. In the present paper we apply the CTRW modeling to the
continuous quantum measurements yielding the new fractional in time evolution equations of quantum
filtering and thus new fractional equations of quantum mechanics of open systems. The related quantum
 control problems and games turn out to be described by the fractional Hamilton-Jacobi-Bellman (HJB)
 equations on Riemannian manifolds. By-passing we provide a full derivation of the standard quantum
 filtering equations, in a modified way as compared with existing texts, which (i) provides
explicit rates of convergence (that are not available via the tightness of martingales approach
developed previously) and (ii)  allows for the direct applications of the basic results of CTRWs to
 deduce the final fractional filtering equations.
\end{abstract}



\section{Introduction}

Direct continuous observations are known to destroy quantum evolutions
 (so-called quantum Zeno paradox), so that continuous quantum measurements have to be indirect,
and the results of the observation are assessed via quantum filtering. Initially developed in the framework
of quantum stochastic calculus by Belavkin in the 80s of the last century in \cite{Bel87},
 \cite{Bel88}, \cite{Bel92}, see \cite{BoutHanJamQuantFilt} for a readable modern account,
the main equations of quantum stochastic filtering, often referred
to as the Belavkin equations, were later on derived
via more elementary approach, as the limit of standard discrete measurements under appropriate scaling,
see e.g. \cite{Be195}, \cite{BelKol}, \cite{Pellegrini}.
The scaling arises from the basic Markovian assumption that the times between measurement are either fixed or
exponentially distributed, like in a standard random walk. Since such Markovian assumption has no a priori
justification, in many branches of modern physics it became popular
to extend random walk modeling to the continuous time random walk (CTRW) modeling, where the time between
discrete events is taken to be non-exponential, usually from the domain of attraction of a stable law.
In the present paper we apply the CTRW modeling to the continuous quantum measurements yielding
 the new fractional in time evolution equations of quantum filtering in the scaling limit. The related quantum
 control problems turn out to be described by the fractional Hamilton-Jacobi-Bellman (HJB) equations on
 Riemannian manifolds (complex projective spaces in the case of finite-dimensional quantum mechanics)
 or the fractional Isaacs equation in the case of competitive control. By-passing
 we provide a full derivation of the standard quantum filtering equations
 (explaining from scratch all underlying quantum mechanical rules used)
  in a slightly modified and simplified way
yielding also new explicit rates of convergence (which are not available via
 the tightness of martingales approach developed previously) and tailored in a way that allows for the direct
 applications of the basic results of CTRWs to deduce the final fractional filtering equations.

Several general comments on a wider context are in order.

(i) The fractional equations of quantum stochastic filtering derived here can be considered as an alternative
formulation of fractional quantum mechanics, which is different from the framework of fractional
Schr\"odinger equations suggested in  \cite{Lask02} and extensively studied recently. This leads also
to a different class of quantum control problems, as those related to fractional Schr\"odinger formulation,
as discussed e. g. in \cite{Wang12}.

(ii) The fractional versions of the classical stochastic filtering (see \cite{BainCrisbook} for the basics)
has been  actively studied recently, see e.g. \cite{Umarov14}.

(iii) The quantum mean-field games as developed by the author in \cite{KolquaMFG} can now be extended to the
theory of fractional quantum mean-field games. The classical versions of fractional mean-field games just
started to appear in the literature, see \cite{Camilli19}.  On the other hand, the application
of classical stochastic filtering in the study of mean-field games has also started to appear, see \cite{CainesFilt}.

(iv) Fractional modeling and CTRW become very popular in almost all domains of physics, as well as economics
and finances, see e.g. \cite{Scalabook}, \cite{MetzlerBark14},
\cite{UchaiBook13}, \cite{WestComplexBook} for some representative references.

The content of the paper is as follows. In Section \ref{secnot} we recall the basic notions
and notations of finite-dimensional quantum mechanics, and in Section \ref{secMarchainindirodserv}
we introduce the Markov chain of sequential indirect quantum measurements, which is the standard starting
point for dealing with continuous measurements. In Sections \ref{seccountcase} and \ref{secdifcase}
we derive the main quantum filtering equations in the cases of so-called counting and diffusive
observations. As was already mentioned, though the derivation of the filtering equations from the approximating Markov chain
is well known by now (see e. g. \cite{Pellegrini10a}) our approach is new and yields
explicit rates of convergence. In Section \ref{secseveralchan} the limiting equation is derived in a general case
of mixed counting and diffusive observations via a multichannel measuring device.
This preparatory work allows us to derive our main results, fractional equations of quantum filtering and control,
in a more or less straightforward way, by applying the established techniques of CTRW to the setting
of the Markov chains of sequential quantum measurements, as developed in Sections \ref{seccountcase}
- \ref{secseveralchan}. This is done in Sections \ref{secfraceq} and \ref{secfraceqHJB}.
In Section \ref{secunbound} we briefly describe a slightly different Markov chain approximation to
continuous measurement that can be used to derive filtering equations in certain cases of unbounded operators
involved.
 In Appendices A,B,C several (known) probabilistic techniques are presented
in a concise form tailored to our purposes. They are used in the main body of the paper.

Some basic notations to be used throughout the text are as follows.

For two Banach spaces $B$ and $D$ equipped with norms $\|.\|_B$ and  $\|.\|_B$ respectively,
let us denote by $\LC(D,B)$ the Banach space of
bounded linear operators in $B$ equipped with the usual operator norm $\|.\|_{D\to B}$.
We shall also write $\LC(B)$ for $\LC(B,B)$.

The scalar product of operators in a Hilbert space is given by the trace:
$(R,S)={\tr} (RS)$.

For $K=\R^d$ or a convex closed subset of $\R^d$ we denote $C(K)$ the Banach space of continuous
bounded functions on $K$, equipped with the sup-norm and $C^k(K)$ the Banach space of $k$ times
continuously differentiable functions on $K$ (with the derivatives at the boundary understood as
the continuous extensions of the derivatives in the inner points), with the norm being the sum
of the sup-norms of the functions and all their partial derivatives of order not exceeding $k$.

\section{Notations for quantum states and tensor products}
\label{secnot}

Recall that a general isolated quantum system is described by a Hilbert space $\HC$ and a self-adjoint operator $H$ in it,
the Hamiltonian. The pure states of the system are unit vectors in $\HC$ and the general mixed states are density matrices,
that is, non-negative operators in $\HC$ with unit trace. Let us denote $S(H)$ the set of all such mixed states in $H$.
 To a pure state there corresponds a density matrix according
to the rule $\psi \to \ga=\psi\otimes \bar \psi$, also denoted in Dirac's notation as $|\psi \rangle \langle \psi|$.
This density matrix is the on-dimensional orthogonal projector on the line generated by $\psi$.
Pure states evolve in time according to the rule $\psi \to e^{-itH} \psi$
and the mixed state according to the rule  $\ga \to e^{-itH} \ga e^{itH}$.

If two systems living in spaces $\HC_0$ and $\HC_1$ are brought to interaction, the combined system
has the tensor product Hilbert space  $\HC_0 \otimes \HC_1$ as the state space. Recall that,
in the coordinate description of tensor products, if $\HC_0$ and $\HC_1$
have orthonormal bases $\{e_j\}$ and $\{f_j\}$ respectively, the tensor product is the space
with an orthnormal basis $\{e_k\otimes f_j\}$. In particular, if $\HC_0$ and $\HC_1$ have finite dimensions $n$ and $k$,
the space $\HC_0 \otimes \HC_1$ has the dimension $nk$.
The operators $A$ in $\HC_0 \otimes \HC_1$ can be  given by matrices $A^{i_1i_2}_{j_1j_2}$, so that
\[
A (e_{i_1}\otimes f_{i_2})=\sum_{j_1,j_2} A_{i_1i_2}^{j_1j_2} e_{j_1}\otimes f_{j_2}.
\]
Or equivalently, if $X\in \HC_0 \otimes \HC_1$ has coordinates $X^{kj}$ in the basis $\{e_k\otimes f_j\}$,
the vector $AX$ has the coordinates $\sum_{m,l} A^{kj}_{ml} X^{ml}$ in this basis.

A product $A\otimes B$ of two operators $A$ and $B$ acting in $\HC_0$ and $\HC_1$ respectively
is defined by its action on tensor products as
\[
(A\otimes B)(e\otimes f)=Ae\otimes Bf.
\]
In the coordinate description $A\otimes B$ has the matrix elements expressed as $A^{i_1}_{j_1}B^{i_2}_{j_2}$
in terms of the matrix elements of $A$ and $B$.

An operator $A$ in $\HC_0$ has the natural lifting $A\otimes I$ (where $I$ is the unit operator) to $\HC_0\otimes \HC_1$.
Similarly an operator $B$ in $\HC_1$ has the natural lifting $I\otimes B$ to $\HC_0\otimes \HC_1$.

The key notion of the theory of interacting systems is that of the {\it partial trace}. For an operator
$A$ in $\HC_0 \otimes \HC_1$ the partial trace with respect to the second system is the operator ${\tr}_{p1} A$
 in $\HC_0$ given by the matrix
 \begin{equation}
 \label{eqdefparttr}
 ({\tr}_{p1} A)^i_j=\sum_k A^{ik}_{jk}.
 \end{equation}
This  partial trace is interpreted as the state of the first system given the state of the coupled one.
Therefore it can be looked at as the quantum analog of the notion of marginal distribution of classical
probability. Similarly, the partial trace with respect to the first system is the operator ${\tr}_{p0} A$
 in $\HC_1$ given by the matrix
 \[
 ({\tr}_{p0} A)^i_j=\sum_k A^{ki}_{kj}.
 \]
 Clearly,
 \[
 {\tr} ({\tr}_{p0} A)= {\tr} ({\tr}_{p1} A)= {\tr} (A).
 \]

In a two-dimensional Hilbert spaces
$\C^2$ one usually chooses the standard basis  $e_0=(1,0)$, $e_1=(0,1)$,
and represents the Hilbert product space $\HC_0\otimes \C^2$ by the natural
decomposition
\[
\HC_0 \otimes \C^2=\HC_{00}\oplus H_{01}=\HC_0 \otimes e_0 \oplus \HC_0 \otimes e_1.
\]
Every operator $A$ in this space has the block decomposition
\[
A=
\begin{pmatrix}
 A_{0\to 0} & A_{1\to 0}  \\

 A_{0\to 1} & A_{1\to 1}
\end{pmatrix}
=
\begin{pmatrix}
 (A_{j0}^{i0}) &  (A_{j1}^{i0}) \\

 (A_{j0}^{i1}) & (A_{j1}^{i1})
\end{pmatrix}
\]
where the operators $A_{i\to j}$ act from $\HC_{0i}$ to $\HC_{0j}$, $i,j=0,1$. The trace
\eqref{eqdefparttr} gets the expression
 \begin{equation}
 \label{eqdefparttrwithtwodim}
 ({\tr}_{p1} A)^i_j=A^{i0}_{j0}+A^{i1}_{j1}.
 \end{equation}

In particular, we shall use the following block representations:
\begin{equation}
\label{eqblockmatrix}
A\otimes I=
\begin{pmatrix} A & 0 \\ 0 & A \end{pmatrix},
\quad A\otimes \Om=
\begin{pmatrix} A & 0 \\ 0 & 0 \end{pmatrix},
\quad
  C\otimes
  \begin{pmatrix} 0 & 0 \\ 1 & 0 \end{pmatrix}
  =\begin{pmatrix} 0 & 0 \\ C & 0 \end{pmatrix},
  \quad
C\otimes
\begin{pmatrix} 0 & 1 \\ 0 & 0 \end{pmatrix}
=\begin{pmatrix} 0 & C \\ 0 & 0 \end{pmatrix}
\end{equation}

More generally,
if $B=(B^i_j)$ is a matrix in $\C^2$, then the matrix of $I\times B$ in $\HC\otimes \C^2$
has the block decomposition
\begin{equation}
\label{eqblockmatrep}
\begin{pmatrix} B^0_0 I & B_1^0 I \\
B^1_0 I & B_1^1 I
\end{pmatrix}.
\end{equation}

To conclude this section let us write down the simple small time asymptotic
formula for the evolutions $e^{-itH}$
that we shall use repeatedly.  Namely, up to the terms
of order higher than $t^2$ in small $t$, we have
\[
e^{-itH}\rho e^{itH}
=(1-it H -\frac12 t^2 H^2)\rho (1+it H-\frac12 t^2 H^2)
=\rho-it [H,\rho]-\frac12 t^2 H^2\rho -\frac12 t^2 \rho H^2+t^2 H\rho H
\]
\begin{equation}
\label{smalltimegroup}
=\rho-it [H,\rho]+t^2 (H\rho H-\frac12 \{H^2,\rho\}).
\end{equation}

\section{The starting point: Markov chains of sequential indirect observations}
\label{secMarchainindirodserv}

Here we describe the
Markov chains of sequential indirect observations (rather standard by now, at least after paper \cite{Attal})
in discrete and continuous time recalling first quickly the main notions related to quantum measurements.

{\it Physical observables} are given by self-adjoint operators $A$ in $\HC$. If $A$ has
a discrete spectrum (which is always the case in finite-dimensional $\HC$, that we shall mostly work with),
then $A$ has the spectral decomposition $A=\sum_j \la_j P_j$, where
 $P_j$ are orthogonal projections on the eigenspaces of $A$ corresponding to
 the eigenvalues $\la_j$. According to the {\it basic postulate of quantum measurement}
 \index{basic postulate of quantum measurement}, measuring observable $A$ in a state $\ga$
(often referred to as the {\it Stern-Gerlach experiment}\index{ Stern-Gerlach experiment})
can yield each of the eigenvalue $\la_j$ with the probability
   \begin{equation}
\label{eqantiunitdress1}
{\tr} \, (\ga P_j)={\tr} \, (P_j \ga P_j),
\end{equation}
 and, if the value $\la_j$ was obtained, the state
of the system changes (instantaneously) to the reduced state
\[
P_j\ga P_j/ {\tr} \, (\ga P_j).
\]
In particular, if the state $\rho$ was pure, $\ga=|\psi\rangle \langle \psi|$, then the
probability to get $\la_j$ as the result of the measurement becomes $(\psi_,P_j\psi)$
and the reduced state also remains pure and is given by the vector $P_j\psi$.
If the interaction with the apparatus was preformed 'without reading the results',
the state $\rho$ is said to be subject to a {\it non-selective measurement}
\index{non-selective measurement} that changes $\ga$ to the state $\sum_j P_j\rho P_j$.

Indirect measurements of a chosen quantum system in the initial space $\HC_0$,
which we shall often referred to as an atom, are organised in the following way.
One couples the atom with another quantum system, a measuring devise, specified
by another Hilbert space $\HC$. Namely the combined system lives in the tensor product
 Hilbert space $\HC_0\times \HC$ and its evolution is given by  certain self-adjoint
 operator $H$ in  $\HC_0\times \HC$. In the measuring device some fixed vector $\varphi \in \HC$
 is chosen, called the vacuum and interpreted as the stationary state of the devise when
 no interaction is involved. The corresponding density matrix will be denoted $\Om=|\varphi \rangle \langle \varphi|$.
 Indirect measurements of the states of the atom
 are performed by measuring the coupled system via an observable of the second system
 and then projecting the resulting state to the atom via the partial trace.

Namely it is described by an operator $R$ in $\HC$ with the spectral decomposition
$R=\sum_j \la_j P_j$ and is performed in two steps: given a state $\ga$ in  $\HC_0\times \HC$
one performs a measurement of $R$ lifted as $I\otimes R$ to $\HC_0\times \HC$
yielding values $\la_j$ and new states
\[
(I\otimes P_j)\ga (I\otimes P_j)/ {\tr} \, (\ga  (I\otimes P_j))
\]
with probabilities $p_j= {\tr} \, (\ga  (I\otimes P_j))$, and then one projects these states to $\HC_0$
via the partial trace producing the states
\begin{equation}
\label{eqindirmeas}
{\tr}_{p1} [(I\otimes P_j)\ga (I\otimes P_j)/ {\tr} \, (\ga  (I\otimes P_j))].
\end{equation}

The discrete time {\it Markov chain of successive indirect observations} (or measurements) evolves according
to the following procedure specified by a triple: a self-adjoint  operator $H$ in $\HC_0\times \HC$,
a self-adjoint operator $R$ in $\HC$ and the vacuum vector $\Om$ in $\HC$.
(i) Starting with an initial state $\rho$ of $\HC_0$ one couples it with the device in its vacuum state $\Om$
producing the state $\ga=\rho\otimes \Om$ in $\HC_0\times \HC$,
(ii) During a fixed period of time $t$ one evolves the system according to the operator $H$ producing the state
$\ga_t=e^{-itH}\ga e^{itH}$ in $\HC_0\times \HC$,
(iii) One performs the indirect measurement with the state $\ga_t$ yielding the states
\begin{equation}
\label{eqMarkchain}
\rho_t^j={\tr}_{p1} \frac{(I\otimes P_j)\ga_t (I\otimes P_j)}{p_j(t)}
={\tr}_{p1} \frac{(I\otimes P_j)e^{-itH}(\rho\otimes \Om) e^{itH} (I\otimes P_j)}{p_j(t)}
\end{equation}
with the probabilities
\begin{equation}
\label{eqMarkchain1}
p_j(t)={\tr} \, (\ga_t  (I\otimes P_j))
={\tr} \, (e^{-itH}(\rho\otimes \Om) e^{itH}  (I\otimes P_j)).
\end{equation}

Then the same repeats starting with $\rho_t$ as the initial state.  Let us denote $U_t$ the transition
operator of this Markov chain that acts on the set of continuous functions on $S(H)$ as
\begin{equation}
\label{eqMarkchain2}
U_t f(\rho)=\E f(\rho_t)=\sum_j p_j(t) f(\rho_t^j).
\end{equation}
Similarly one can define the continuous time {\it Markov chain of successive indirect observations} (or measurements)
$O^{\rho}_{t,\la}$ and the corresponding Markov semigroup $T_t^{\la}$ on $C(H(S))$
evolving according to the same rules, with only difference that the times $t$ between successive measurements
are not fixed, but represent exponential random  variables $\tau$ with some fixed intensity $\la$: $\P(\tau>t)=e^{-\la t}$.
The generator $L^{\la}$ of this Markov process is bounded in $C(S(H))$ and acts as
\begin{equation}
\label{eqMarkchain3}
L^{\la}f(\rho)=\frac{(U_{\la}f-f)(\rho)}{\la}=\frac{1}{\la}\sum_j p_j(t) (f(\rho_t^j)-f(\rho)).
\end{equation}

All "quantum content" of the theory is now captured in the explicit formula \eqref{eqMarkchain}.
What follows will be
the pure classical probability analysis of these Markov chains, their scaling limits and control.

In this paper we shall work with the measuring devises of the simplest form living in two-dimensional Hilbert spaces
$\C^2$ or more generally the tensor products of these spaces. Choosing the standard basis $e_0=(1,0)$, $e_1=(0,1)$,
we shall use the decomposition
\[
\HC_0\otimes \C^2=\HC_{00}\oplus H_{01}=\HC_0 \otimes e_0 \oplus \HC_0 \otimes e_1,
\]
and we shall choose the vacuum vector $\varphi=e_0$, so that
\[
\Om=\begin{pmatrix}
 1 & 0  \\
 0 & 0
\end{pmatrix}
\]

\section{Belavkin equations for a counting observation}
\label{seccountcase}

For simplicity we shall work exclusively with finite-dimensional Hilbert spaces $\HC_0=\C^n$,
making occasionally some comments about more general case.
The set of states $S(\C^n)$ is a compact convex set in the Euclidean space $\R^{n^2}$,
the space of complex Hermitian $n\times n$ matrices.

Let us choose an arbitrary self-adjoint operator in $\HC_0\otimes \C^2$ given by its matrix representation
\[
H= \begin{pmatrix} A & 0 \\ 0 & B \end{pmatrix}
+  \begin{pmatrix} 0 & -iC^* \\ iC & 0 \end{pmatrix}.
\]

We are aiming at calculating the small time asymptotics of the Markov transition operators defined by
\eqref{eqMarkchain}.

The main idea for obtaining sensible asymptotic limits suggests enhancing the interaction part $C$ of $H$ by
 replacing it with the scaled version $C/\sqrt t$. Thus we choose the Hamiltonian in the form

 \[
H= \begin{pmatrix} A & 0 \\ 0 & B \end{pmatrix}
+ \frac{1}{\sqrt t} \begin{pmatrix} 0 & -iC^* \\ iC & 0 \end{pmatrix}.
\]

 \begin{remark}
 The idea of the scaling comes from the analysis of the so-called quantum Zeno paradox.
 Its essence is a rather simple observation that if one performs repeated measurements with reduction
 \eqref{eqantiunitdress1} and pass to the limit,
  as time between measurements tends to zero, then the state effectively remains in the initial
  state all the time irrespectively of the dynamics. This effect is also referred to as the watch dog effect.
  Therefore the only way to get a sensible dynamics that takes into account both dynamics and observation is
  to enhance the interaction part of the dynamics to make its effect comparable with that of the repeated reduction
  \eqref{eqantiunitdress1}. Thus one can suggest scaling $C$ as $C/t^{\al}$ with some $\al>0$. As calculations show
  (one can repeat the calculations below with an arbitrary $\al$)
  only with $\al=1/2$ a sensible limit is obtained.
 \end{remark}

By the second equation in \eqref{eqblockmatrix}, we get

\[
\rho \otimes \Om =\begin{pmatrix} \rho & 0 \\ 0 & 0 \end{pmatrix},
\quad
\left[H, \begin{pmatrix} \rho & 0 \\ 0 & 0 \end{pmatrix}\right]
=\begin{pmatrix} [A,\rho] & + i\rho C^*/\sqrt t \\ iC\rho/\sqrt t & 0 \end{pmatrix}
\]
\[
H\begin{pmatrix} \rho & 0 \\ 0 & 0 \end{pmatrix}H
=\begin{pmatrix} A\rho A & -i A\rho C^*/\sqrt t\\ iC\rho A/\sqrt t & C\rho C^*/t \end{pmatrix},
\quad
H^2 =\begin{pmatrix} A^2+C^*C/t & -i(AC^* +C^* B)/\sqrt t \\ i(CA+BC)/\sqrt t & B^2+CC^*/t \end{pmatrix},
\]
\[
\{H^2, \rho \otimes \Om\}=\begin{pmatrix} \{A^2+C^*C/t,\rho\} & -i\rho (AC^*+C^*B)/\sqrt t \\ i(CA+BC)\rho/\sqrt t  & 0 \end{pmatrix}.
\]
where $\{C,D\}=CD+DC$ denotes the anti-commutator.
Using \eqref{smalltimegroup}, and keeping terms of order not exceeding $t$ we get the approximation
\begin{equation}
\label{eqdressedrho1}
e^{-itH} (\rho\otimes \Om) e^{itH}
=\begin{pmatrix} \rho -it[A, \rho] -\frac12 t\{C^*C,\rho\} & \sqrt t \rho C^* \\ \sqrt t C\rho & tC\rho C^* \end{pmatrix},
\end{equation}
which is the key formula for what follows.

As it turns out, the limiting processes are of two types, depending on whether
the projectors $P_0$ and $P_1$ of the spectral decomposition of $R$ are diagonal, that is
\begin{equation}
\label{eqdiagonproj}
P_0= \begin{pmatrix} 1 & 0 \\ 0 & 0 \end{pmatrix}, \quad P_1=\begin{pmatrix} 0 & 0 \\ 0 & 1 \end{pmatrix}
\end{equation}
or otherwise. Let us start with the case of projectors \eqref{eqdiagonproj}.

We have
\[
I\otimes P_0= \begin{pmatrix} I & 0 \\ 0 & 0 \end{pmatrix},
\quad I\otimes P_1 =\begin{pmatrix} 0 & 0 \\ 0 & I \end{pmatrix},
\]
and
\[
(I\otimes P_0) e^{-itH} \begin{pmatrix} \rho & 0 \\ 0 & 0 \end{pmatrix} e^{itH} (I\otimes P_0)
= \rho -it[A, \rho] -\frac12 t\{C^*C,\rho\},
\]
\[
(I\otimes P_1) e^{-itH} \begin{pmatrix} \rho & 0 \\ 0 & 0 \end{pmatrix} e^{itH} (I\otimes P_1)
= tC\rho C^*.
\]
Hence the non-normalized new states are
\[
\tilde \rho_1=\rho -it[A, \rho] -\frac12 t\{C^*C,\rho\},
\quad
\tilde \rho_2= tC\rho C^*,
\]
occurring with the probabilities
\[
p_1=1-t \, {\tr} (C^*C \rho), \quad p_2=t \, {\tr} (C^*C \rho).
\]

Aiming at using Proposition \ref{propconvsemigr} (ii) we are looking for
the limit of the operator $(U_h-1)/h$ for $h \to 0$.

Denoting $T=  {\tr} (C^*C \rho)$ we can write up to terms of order $t$ that
\[
\frac{U_h-1}{h} f(\rho)=\frac{1}{h} (1-hT)\left[f(\frac{\tilde \rho_1}{1-hT})-f(\rho)\right]
+\frac{1}{h} h\, T \left[(f(\frac{\tilde \rho_2}{hT})-f(\rho))\right]
\]
\[
\approx  \frac{1}{h} (1-hT)[f(\rho -it[A, \rho] -\frac12 t\{C^*C,\rho\}+h\rho T)-f(\rho)]
+T \left[f(\frac{C\rho C^*}{T})-f(\rho)\right]\approx L_{count}f,
\]
with
\begin{equation}
\label{eqjumpgener}
L_{count}f(\rho)=-(f'(\rho), i[A, \rho] +\frac12 \{C^*C,\rho\}-\rho T)
+T \left[f(\frac{C\rho C^*}{T})-f(\rho)\right].
\end{equation}

Summarising by looking carefully at the small terms ignored, we can conclude the following.

\begin{lemma}
\label{lemmaongencount}
Under the setting considered,
\begin{equation}
\label{eqjumpgener1}
\|\frac{U_h-1}{h} f -Lf\|\le \sqrt h \ka \|f\|_{C^2(S(\HC_0))}
\end{equation}
for $f\in C^2(S(\HC_0))$, with $L_{count}$ given by \eqref{eqjumpgener} and a constant $\ka$.
\end{lemma}

We can prove now our first result.

\begin{theorem}
\label{thBeleqcount}
Let $\HC_0=\C^n$ and $A$, $C$ be $n\times n$ square matrices with $A$ being Hermitian. Then:

(i) The operator  \eqref{eqjumpgener} generates a Feller process $O_t^{\rho}$ in $S(\HC_0)$ and
the corresponding Feller semigroup $T_t$ in $C(S(\HC_0))$ having the spaces
 $C^1(S(\HC_0))$ and $C^2(S(\HC_0))$ as invariant cores, and $T_s$ are bounded in these spaces
 uniformly for $s\in [0,t]$ with any $t>0$.

(ii) The scaled discrete semigroups $(U_h)^{[s/h]}$ converge to the semigroup $T_s$, as $h\to 0$, so that the corresponding
processes converge in distribution, with the following rates of convergence:
\begin{equation}
\label{eq1thBeleqcount}
\|(U_h)^{[s/h]} -T_sf\| \le \sqrt h s \ka(t) \|f\|_{C^2(S(\HC_0))},
\end{equation}
where the constant $\ka(t)$ depends on the dimension $n$ and the norms of $A$ and $C$.

(iii) The scaled semigroups $T_s^{\la}$ converge to the semigroup $T_s$, as $\la\to 0$, so that the corresponding
processes converge in distribution, with the following rates of convergence:
\begin{equation}
\label{eq2thBeleqcount}
\|T_s^{\la}f -T_sf\| \le \sqrt \la s \ka(t) \|f\|_{C^2(S(\HC_0))}.
\end{equation}
\end{theorem}

\begin{proof}
(i) This is a consequence of Proposition  \ref{propdetjumpdom}. To make this conclusion one needs to show
property \eqref{eqBrezis} with $K=S(\C^n)$ and
\[
b(\rho)= -i[A, \rho] -\frac12 \{C^*C,\rho\}+ {\tr} (C^*C\rho) \rho.
\]
It is straightforward to see that the solutions to the ODE $\dot \rho=b(\rho)$ preserve the affine set of
Hermitian matrices with unit trace. So the key point is the preservation of positivity. It turns out
that a stronger version of \eqref{eqBrezis} holds, namely that
$d(\rho +h b(\rho),K)=0$
for any $\rho$ from the boundary of $K$ and all sufficiently small $h$. By the compactness of a unit ball in $\C^n$,
this claim follows from the following one. If $\rho$ belongs to the boundary of $K$,
that is, there exists a nonempty set $V(\rho)$ of unit vectors such that $\rho v=0$ for $v\in V(\rho)$,
then $(v,(\rho +h b(\rho))v)\ge 0$ for any unit vector $v$ and all sufficiently small $h$. But this property
is obvious for $v\notin V(\rho)$. On the other hand $(v,b(\rho)v)=0$ for $v\in V(\rho)$ implying that
$(v,(\rho +h b(\rho))v)= 0$ for all $h>0$ and all $v\in V(\rho)$.

(ii) This is a consequence of (i), Proposition \ref{propconvsemigr} (ii) and the observation that \eqref{eq5propconvsemigr}
holds here with  the triple of spaces $C^2(S(\HC_0))\subset C^1(S(\HC_0)) \subset C(S(\HC_0))$.

(iii)  This is a consequence of (i), formula \eqref{eqMarkchain3} and Proposition \ref{propconvsemigr} (i),
with $B=C(S(\HC_0))$, $D=C^2(S(\HC_0))$.
\end{proof}

\begin{remark}
\label{remoninfindim}
 This result extends almost automatically to the case of an arbitrary separable Hilbert space $\HC_0$
and arbitrary bounded operators $H,C$, with the derivatives understood in the Fr\'echet sense. The only point
where the finite-dimensional setting was used was in proving statement (i) using compactness of a unit ball in
 $\C^n$ and the Brezis theorem. In infinite-dimensional case one can use the compactness of a unit ball in
 a Hilbert space in the weak topology and the Banach-space version of the Brezis theorem, as presented in
 \cite{Martin} and \cite{Lakshm}.
 \end{remark}

As is seen directly via Ito's formula, the Feller process $O_t^{\rho}$ generated by \eqref{eqjumpgener} can be
described as solving the jump type SDE
\begin{equation}
\label{eqBeleqcount}
d\rho=(- i[A, \rho] -\frac12 \{C^*C,\rho\}+ {\tr} (C\rho C^*) \rho ) dt
+\left(\frac{C\rho C^*}{{\tr} (C\rho C^*)}-\rho\right) dN_t,
\end{equation}
with the counting process $N_t$ with the position dependent intensity ${\tr} (C^*C\rho)$,
so that the compensated process $N_t-\int_0^t {\tr} (C^*C\rho_s) \, ds$ is a martingale.
Equation \eqref{eqBeleqcount} is the {\it Belavkin quantum filtering SDE} corresponding
to the {\it counting type observation} (because the driving process $N_t$ is a counting process).
Representation via the generator is an equivalent way of specifying the process of
continuous quantum observation and filtering.

\begin{remark} Equation \eqref{eqBeleqcount} is slightly nonstandard as the driving noise $N_t$ is itself
position dependent. However there is a natural way to rewrite it in terms of an independent driving noise.
Namely, with a standard Poisson random measure process $N(dx \,dt)$ on $\R_+\times \R_+$ (with Lebesgue measure
as intensity) one can rewrite equation \eqref{eqBeleqcount} in the following equivalent form:
\begin{equation}
\label{eqBeleqcountindepnoise}
d\rho=(- i[A, \rho] -\frac12 \{C^*C,\rho\}+ {\tr} (C\rho C^*) \rho ) dt
+\left(\frac{C\rho C^*}{{\tr} (C\rho C^*)}-\rho\right) \1({\tr} (C^*C\rho)\le x) N(dx\, dt),
\end{equation}
see details of this construction in \cite{Pellegrini10a}. Alternatively, one can make sense
of \eqref{eqBeleqcount} in terms of the general theory of weak SDEs from \cite{Kol11}.
\end{remark}

\begin{remark} The meaning of the term 'counting observation' (as well as 'diffusive type' of the next section)
becomes more concrete in a more advanced
treatment of the process of quantum measurement, see e.g.  \cite{BoutHanJamQuantFilt}.
\end{remark}

\section{Belavkin equations for a diffusive observation}
\label{secdifcase}

Let us turn to the second case of choosing orthogonal projectors $P_0,P_1$,
when they differ from the diagonal choice \eqref{eqdiagonproj}.

General couple of two orthogonal projectors in $\C^2$ is easily seen to be of the form
\[
P_0=\begin{pmatrix} \cos^2 \phi & \sin\phi \cos \phi e^{i\psi}\\  \sin\phi \cos \phi e^{i\psi} & \sin^2\phi \end{pmatrix},
\quad
P_1=\begin{pmatrix} \sin^2 \phi & -\sin\phi \cos \phi e^{i\psi}\\  -\sin\phi \cos \phi e^{i\psi} & \cos^2\phi \end{pmatrix}.
\]

The phase terms with $\psi$ does not make much difference, so we choose further $\psi=0$. Moreover, to avoid
diagonal case we assume $\phi \neq \pi k/2$, $k\in N$.

By \eqref{eqblockmatrep},
\[
I\times P_0
= \begin{pmatrix} \cos^2 \phi I & \sin\phi \cos \phi I \\  \sin\phi \cos \phi I & \sin^2\phi I\end{pmatrix},
\quad
I\times P_1
= \begin{pmatrix} \sin^2 \phi I & -\sin\phi \cos \phi I \\  -\sin\phi \cos \phi I & \cos^2\phi I\end{pmatrix}.
\]

Hence, for arbitrary matrices $a,b,c,d$, we have
\[
(I\times P_0) \begin{pmatrix} a & b \\  c & d \end{pmatrix}
= \begin{pmatrix} \cos^2 \phi \, a + \sin\phi \cos \phi \, c & \cos^2 \phi \, b + \sin\phi \cos \phi \, d \\
 \sin\phi \cos \phi \, a + \sin ^2\phi \, c & \sin\phi \cos \phi \, b + \sin^2\phi \, d\end{pmatrix}
 \]
 and
 \[
 (I\times P_0) \begin{pmatrix} a & b \\  c & d \end{pmatrix} (I\times P_0)
 =\begin{pmatrix} \cos^2 \phi \, \om_{\phi} &  \sin\phi \cos \phi  \, \om_{\phi}   \\
   \sin\phi \cos \phi  \, \om_{\phi}  &  \sin^2\phi \, \om_{\phi}  \end{pmatrix}
   \]
   with
   \[
   \om_{\phi}=\om_{\phi}(a,b,c,d)=
   \cos^2 \phi \, a+  \sin\phi \cos \phi (b+c)+ \sin^2\phi \, d.
   \]
   Since $P_1$ is obtained from $P_0$ by changing $\phi$ to $\phi+\pi/2$, it follows that
\[
 (I\times P_1) \begin{pmatrix} a & b \\  c & d \end{pmatrix} (I\times P_1)
 =\begin{pmatrix} \sin^2 \phi \, \tilde \om_{\phi} &  -\sin\phi \cos \phi  \, \tilde \om_{\phi}   \\
   -\sin\phi \cos \phi  \, \tilde \om_{\phi}  &  \cos^2\phi \, \tilde \om_{\phi}  \end{pmatrix}
   \]
   with
   \[
   \tilde \om_{\phi}=\om_{\phi+\pi/2}
  = \sin^2 \phi \, a -\sin\phi \cos \phi (b+c)+ \cos^2\phi \, d.
   \]

By \eqref{eqdefparttrwithtwodim} we get
\[
{\tr}_{p1} [(I\times P_0) \begin{pmatrix} a & b \\  c & d \end{pmatrix} (I\times P_0)]
\]
\[
=\om_{\phi} = \cos^2 \phi \, a+  \sin\phi \cos \phi (b+c)+ \sin^2\phi \, d,
\]
\[
{\tr}_{p1} [(I\times P_1) \begin{pmatrix} a & b \\  c & d \end{pmatrix} (I\times P_1)]
\]
\[
=\tilde \om_{\phi} = \sin^2 \phi \, a-  \sin\phi \cos \phi (b+c)+ \cos^2\phi \, d.
\]

To get new states we have to take $a,b,c,d$ from \eqref{eqdressedrho1}.
Hence for the non-normalized states we get the approximate formulas (up to terms of order $t$):
\[
\tilde \rho_1 = \cos^2 \phi (\rho -it[A, \rho] -\frac12 t\{C^*C,\rho\})
+  \sqrt t \sin\phi \cos \phi (\rho C^* + C\rho)+ t \sin^2\phi \,C\rho C^*,
\]
\[
\tilde \rho_2 = \sin^2 \phi (\rho -it[A, \rho] -\frac12 t\{C^*C,\rho\})
- \sqrt t \sin\phi \cos \phi (\rho C^* + C\rho)+ t\cos^2\phi \, C\rho C^*.
\]
These states occur
with the probabilities
\[
p_1=\cos^2\phi(1-tT)+\sqrt t \sin\phi \cos \phi \, {\tr} (\rho C^* + C\rho)+t T \sin^2\phi,
\]
\[
p_2=\sin^2\phi(1-tT)-\sqrt t \sin\phi \cos \phi \, {\tr} (\rho C^* + C\rho)+t T \cos^2\phi.
\]

For arbitrary numbers $a,b,c$, one can write up to terms of order $t$, that
\[
\frac{1}{a+b\sqrt t +ct}
=\frac{1}{a} \frac{1}{1+(b/a) \sqrt t+(c/a) t}
=\frac{1}{a}(1-(b/a) \sqrt t-(c/a) t+(b/a)^2 t).
\]
Consequently, with this order of approximation,
\[
\frac{1}{p_1}=\frac{1}{\cos^2 \phi}(1- \tan \phi \sqrt t \, {\tr} (\rho C^* + C\rho)-T(\tan^2\phi-1) t
+\tan^2 \phi \, [{\tr} (\rho C^* + C\rho)]^2 t),
\]
\[
\frac{1}{p_2}=\frac{1}{\sin^2 \phi}(1+ \cot \phi \sqrt t \, {\tr} (\rho C^* + C\rho) -T(\cot^2\phi-1) t
+\cot^2 \phi \, [{\tr} (\rho C^* + C\rho)]^2 t),
\]
and therefore the normalized states are given by the formulas
\[
\rho_1=\frac{\tilde \rho_1}{p_1}
= [\rho -it[A, \rho] -\frac12 t\{C^*C,\rho\}
+  \sqrt t \tan\phi (\rho C^* + C\rho)+ t \tan^2\phi \,C\rho C^*]
\]
\[
\times(1- \tan \phi \sqrt t \, {\tr} (\rho C^* + C\rho)-T(\tan^2\phi-1) t
+\tan^2 \phi \, [{\tr} (\rho C^* + C\rho)]^2 t)
\]
\[
=\rho +  \sqrt t \tan\phi (\rho C^* + C\rho-\Om \rho)+tB_1
\]
with
\[
\Om= {\tr} (\rho C^* + C\rho)
\]
and
\[
B_1=
 -i[A, \rho] -\frac12 \{C^*C,\rho\}+ T\rho
 +\tan^2\phi (C\rho C^*- (\rho C^* + C\rho) \Om -T\rho+ \Om^2  \rho),
\]

and
\[
\rho_2=\frac{\tilde \rho_2}{p_2}
= [\rho -it[A, \rho] -\frac12 t\{C^*C,\rho\}
-  \sqrt t \cot\phi (\rho C^* + C\rho)+ t \cot^2\phi \,C\rho C^*]
\]
\[
\times(1+ \cot \phi \sqrt t \, {\tr} (\rho C^* + C\rho)-T(\cot^2\phi-1) t
+\cot^2 \phi \, [{\tr} (\rho C^* + C\rho)]^2 t)
\]
\[
=\rho - \sqrt t \cot\phi (\rho C^* + C\rho-\Om\rho)+tB_2
\]
with
\[
B_2=
-i[A, \rho] -\frac12 \{C^*C,\rho\}+T\rho
+  \cot^2\phi (C\rho C^*- (\rho C^* + C\rho) \Om -T\rho+\Om^2  \rho).
\]

The terms of order $t$ in $p_j$ give contributions of lower order, so that
to the main order in small $h$ we have
\[
\frac{U_h-1}{h} f(\rho)=\frac{1}{h} p_1\left[f(\rho_1)-f(\rho)\right]
+ \frac{1}{h} p_1\left[f(\rho_2)-f(\rho)\right].
\]
\[
=\frac{1}{h}(\cos^2\phi+\sqrt h \sin\phi \cos \phi \Om)
\]
\[
\times \left[(f'(\rho),\sqrt h \tan\phi (\rho C^* + C\rho-\Om \rho)+tB_1)
+\frac12 \tan^2\phi [(\rho C^* + C\rho-\Om \rho)f''(\rho) (\rho C^* + C\rho-\Om \rho)]h\right]
\]
\[
+\frac{1}{h}(\sin^2\phi-\sqrt h \sin\phi \cos \phi \Om)
\]
\[
\times \left[ (f'(\rho),- \sqrt h \cot\phi (\rho C^* + C\rho-\Om\rho)+tB_2)
+\frac12 \cot^2\phi [(\rho C^* + C\rho-\Om \rho)f''(\rho) (\rho C^* + C\rho-\Om \rho)]h\right],
\]
where
\[
[Af''(\rho) A]=\sum_{ijkl} A_{ij} \frac{\pa^2 f}{\pa \rho_{ij}\pa \rho_{kl}} A_{kl}.
\]
The terms of order $h^{-1/2}$ cancel and we get in the main term
\[
\frac{U_h-1}{h} f(\rho)\approx \frac12 [(\rho C^* + C\rho-\Om \rho)f''(\rho) (\rho C^* + C\rho-\Om \rho)]
\]
\[
+(f'(\rho),\Om (\rho C^* + C\rho-\Om \rho)+\cos^2\phi B_1+\sin^2\phi B_2)\approx L_{dif}f(\rho)
\]
with
\begin{equation}
\label{eqdifgener}
L_{dif}f(\rho)=\frac12 [(\rho C^* + C\rho-\Om \rho)f''(\rho) (\rho C^* + C\rho-\Om \rho)]
+(f'(\rho),-i[A, \rho] -\frac12 \{C^*C,\rho\}+ C\rho C^* ),
\end{equation}
which is remarkably independent of $\phi$! Thus, taking into account the terms that were ignored
withinn the approximation, we obtained the following counterpart of Lemma \ref{lemmaongencount}:

\begin{lemma}
\label{lemmaongendif}
Under the setting considered, and for any $\phi\neq \pi k/2$, $k\in \Z$,
\begin{equation}
\label{eqdifgener1}
\|\frac{U_h-1}{h} f -Lf\|\le \sqrt h \ka \|f\|_{C^3(S(\HC_0))}
\end{equation}
for $f\in C^3(S(\HC_0))$, with $L_{dif}$ given by \eqref{eqdifgener}.
\end{lemma}

Unlike the jump-type limiting processes analysed in the previous section,
where a straightforward pure analytic proof of the well-posedness of the process generated by $L$ is available,
here an approach using SDEs is handy. Ito's formula shows that  a process generated by \eqref{eqdifgener}
can arise from solving the following Ito's SDE:

\begin{equation}
\label{eqBeleqdif}
d\rho=(-i[A, \rho] -\frac12 \{C^*C,\rho\}+ C\rho C^* ) \, dt
+\left(\rho C^* + C\rho-{\tr}\, (\rho C^* + C\rho)\rho\right) \, dW_t,
\end{equation}
where $W_t$ is a standard one-dimensional Wiener process.
This SDE is the {\it Belavkin quantum filtering SDE} for normalized states corresponding
to the {\it diffusive type observation}.

\begin{theorem}
\label{thBeleqdif}
Let $\HC_0=\C^n$ and $A$, $C$ be $n\times n$ square matrices with $A$ being Hermitian. Then:

(i) The operator  \eqref{eqdifgener} generates a Feller process $O_t^{\rho}$ in $S(\HC_0)$ and
the corresponding Feller semigroup $T_t$ in $C(S(\HC_0))$ having the spaces
 $C^2(S(\HC_0))$ and $C^3(S(\HC_0))$ as invariant cores, and $T_s$ are bounded in these spaces
 uniformly for $s\in [0,t]$ with any $t>0$. This process is given by the solutions to SDE
 \eqref{eqBeleqdif}, which is well posed as a diffusion equation in $S(\HC_0)$.

(ii) The scaled discrete semigroups $(U_h)^{[s/h]}$ converge to the semigroup $T_s$, as $h\to 0$, so that the corresponding
processes converge in distribution, with the following rates of convergence:
\begin{equation}
\label{eq1thBeleqdif}
\|(U_h)^{[s/h]} -T_sf\| \le \sqrt h s \ka(t) \|f\|_{C^3(S(\HC_0))},
\end{equation}
where the constant $\ka(t)$ depends on the norms of $A$ and $C$.

(iii) The scaled semigroups $T_s^{\la}$ converge to the semigroup $T_s$, as $\la\to 0$, so that the corresponding
processes converge in distribution, with the following rates of convergence:
\begin{equation}
\label{eq2thBeleqdif}
\|T_s^{\la}f -T_sf\| \le \la \sqrt s \ka(t) \|f\|_{C^3(S(\HC_0))}.
\end{equation}
\end{theorem}

\begin{proof} Parts (ii) and (iii) are obtained by the same arguments as in the proof of Theorem \ref{thBeleqcount}.
One only has to mention that estimate \eqref{eq3propconvsemigr} needed to apply Proposition \ref{propconvsemigr}
follows from the standard fact of the theory of diffusion that $\E((X_t(x)-x)^2)\le Ct$ for any diffusion $X_t(x)$
with bounded smooth coefficients.
So we need only to prove (i). All claims follow if one can construct a diffusion in $S(\HC_0)$ solving
\eqref{eqBeleqdif}, because in $S(\HC_0)$ all coefficients are bounded, and then both the uniqueness of solution and the
required smoothness of solutions with respect to initial data follow automatically from the smoothness of
the coefficients by the standard tools of Ito's SDEs.
The main difficulty here lies in proving that solutions to \eqref{eqBeleqdif} preserve the set of positive matrices.
But the fact that SDE \eqref{eqBeleqdif} is well-posed in $S(\HC_0)$ is a well known fact, see e.g.
Section 3.4.1 in monograph \cite{BarchBook}. Thus one can complete a proof of Theorem \ref{thBeleqdif}
by referring  to this result.
However, a proof of \cite{BarchBook} is indirect, and the fact is really  crucial.
 Therefore, for completeness we sketch below a different direct proof that the
solutions to  \eqref{eqBeleqdif} preserve the set of positive matrices.
 In this approach we shall consider the coefficients of the equation
\eqref{eqBeleqdif}  to be given as they are only for nonnegative $\rho$ of unit trace and continued smoothly to all
Hermitian $\rho$ in such a way that these coefficients vanish outside some neighborhood of this set. The modified equations
\eqref{eqBeleqdif} have globally bounded smooth coefficients and hence have unique well defined global solutions. Thus
we really only need to show the preservation of positivity.

Our method  is based on the Stratonovich integral. Recall that the Stratonovich differential $\circ dX$
is lined with Ito's differential by the formula $Z\circ dX =Z \, dX +(1/2)dZ \, dX$. Hence
denoting
\[
B(\rho)=\rho C^* + C\rho-{\tr}\, (\rho C^* + C\rho)\rho,
\]
equation \eqref{eqBeleqdif} rewrites in Stratonovich form as
\[
d\rho=(-i[A, \rho] -\frac12 \{C^*C,\rho\}+ C\rho C^* ) \, dt
+B(\rho) \circ dW_t-\frac12 dB(\rho) \,dW_t
\]
\[
=(-i[A, \rho] -\frac12 \{C^*C,\rho\}+ C\rho C^* ) \, dt
+B(\rho) \circ dW_t
\]
\[
-\frac12 [B(\rho) C^*+CB(\rho)-{\tr} \,(\rho C^*+C\rho) B(\rho)
- {\tr}\,(B(\rho)C^*+CB(\rho))\rho] dt.
\]
Using the fundamental result of the Stratonovich integral, stating that solutions to
Stratonovich SDEs can be obtained as the limits of the solutions to the ODEs obtained
by approximating the white noise with smooth functions, we can state that the
solutions to this Stratonovich equation preserve positivity of matrices, if
the equations
\[
\dot \rho
=-i[A, \rho] -\frac12 \{C^*C,\rho\}+ C\rho C^*
+B(\rho) \phi_t
\]
\begin{equation}
\label{eq2thBeleqdifStappr}
-\frac12 [B(\rho) C^*+CB(\rho)-{\tr} \,(\rho C^*+C\rho) B(\rho)
- {\tr}\,(B(\rho)C^*+CB(\rho))\rho]
\end{equation}
preserve the set of positive matrices for any continuous function $\phi_t$.
But this follows by the Brezis Theorem \ref{thBrezis}.
To see this we substitute the expression for $B(\rho)$ in the first three places of the
last square bracket yielding the equation
\[
\dot \rho
=-i[A, \rho] -\frac12 \{C^*C,\rho\}
+B(\rho) \phi_t
\]
\begin{equation}
\label{eq2thBeleqdifStappr1}
-\frac12 [\rho (C^*)^2+C^2\rho +({\tr} \,(\rho C^*+C\rho))^2 \rho
- {\tr}\,(B(\rho)C^*+CB(\rho))\rho]
\end{equation}
(the key point is that the 'nasty' term $C\rho C^*$ cancels).
It is seen that Theorem \ref{thBrezis} applies, because whenever $(v,\rho v)=0$, the r.h.s.
$\om_t(\rho)$ of equation \eqref{eq2thBeleqdifStappr1} satisfies $(v,\om_t(\rho) v)=0$ for
any function $\phi_t$. The details of the argument are the same as in the proof of Theorem \ref{thBeleqcount}.

\end{proof}

\begin{remark} The methods developed can be used to extend this result to infinite dimensional $\HC_0$.
However, unlike the situation with counting observations, explained in Remark \ref{remoninfindim},
there is some subtlety here in working with SDEs in the space of trace class operators, which are
not going to discuss in this paper.
\end{remark}

A remarkable property of the SDEs \eqref{eqBeleqcount} and \eqref{eqBeleqdif} is that they preserve the pure states.
Namely if the initial state $\rho$ was pure, $\rho=\psi\otimes \bar \psi$, then it remains pure for all times.
Namely, one can check by a direct application of Ito's formula that if $\phi$ satisfies the SDE
\begin{equation}
\label{eqqufiBnonlin}
d\phi=-[i(A-\langle Re \, C\rangle_{\phi} \, Im \, C)
+\frac12(LC-\langle Re \, C\rangle_{\phi})^*(C-\langle Re \, C\rangle_{\phi})]\phi \, dt
+ (C-\langle Re \, C\rangle_{\phi})\phi \, dW_t,
\end{equation}
then $\rho=\psi\otimes \bar \psi$ satisfies equation \eqref{eqBeleqdif}. Equation \eqref{eqqufiBnonlin}
is the Belavkin quantum filtering equation for pure states. It looks much simpler for the most important case
of self-adjoint $C$:
\begin{equation}
\label{eqqufiBnonlins}
d\phi=-[iA+\frac12(L-\langle C\rangle_{\phi})^2]\phi \, dt
+ (C-\langle C\rangle_{\phi})\phi \, dW_t.
\end{equation}

Another key observation is that there exists an equivalent linear version of \eqref{eqBeleqdif}.
Namely assume that $\xi$ solves the following {\it Belavkin quantum filtering SDE} for non-normalized states:
\begin{equation}
\label{eqBeleqdiflin}
d\xi=(-i[A, \xi] -\frac12 \{C^*C,\xi\}+ C\xi C^* ) \, dt
+(\xi C^* + C\xi) \, dY_t,
\end{equation}
where $Y_t$ is a Brownian motion under a certain measure. Applying Ito's formula to $\rho=\xi/{\tr} \, \xi$ one finds that
$\rho$ satisfies \eqref{eqBeleqdif} with the process $W$ satisfying the equation
\begin{equation}
\label{eqBeleqdiflinlink}
dW_t=dY_t-{\tr} \,(\xi C^* + C\xi) \, dt.
\end{equation}

It follows from the famous Girsanov formula that if $Y_t$ was a Wiener process, then $W_t$ would be also a Wiener
process under some different but equivalent measure with respect to one defining $Y_t$. Hence a solution $\xi_t$ to
 the linear equation \eqref{eqBeleqdiflin} with some Brownian motion $Y_t$ yields
the solution $\rho=\xi/{\tr} \, \xi$ to \eqref{eqBeleqdif} with some other Brownian motion $W_t$.

\section{Observations via different channels}
\label{secseveralchan}

Let us now extend the theory to the case of several channels of observation. Namely,
we take
\begin{equation}
\label{interHam}
\HC=\HC_0\otimes \C^2 \otimes \cdots \otimes \C^2, \quad (K \, \text{multipliers}\, \C^2),
\end{equation}
and the atom (system with Hilbert space $\HC_0$) is supposed to interact with each of the $K$
measuring devices with the state space $\C^2$. Each of the devises is equipped with
the standard basis $(e_0^j,e_1^j)$ with $e_0^j$ chosen as a vacuum vector, that is as its stationary state,
with the corresponding density matrix being $\Om_j= |e_0^j\rangle \langle e_0^j|$.
The Hamiltonian is given by the sum $H=H_0+\sum_{k=1}^KH_k$, where
$H_0=A\otimes I^{\otimes k}$ describes the free dynamics of the atom, and $H_j$ connects the atom
with the $j$th device. The same scaling $1/\sqrt t$ applies to the interaction parts.

Thus  $H$ is specified by $k+1$ operators $A,C_1, \cdots, C_K$ in $\HC_0$, so that $H_j$ are
give by the formulas:
 \begin{equation}
\label{interHam1}
\begin{aligned}
& H_0 (h\otimes e_{i_1}^1 \otimes \cdots \otimes e_{i_K}^K)=A h \otimes e^1_{i_1} \otimes \cdots \otimes e_{i_K}^K, \\
& H_j (h\otimes e_{i_1}^1 \otimes \cdots \otimes e_{i_K}^K)|_{e_{i_j}^j=e_1^j}
=-\frac{i}{\sqrt t} C_j^* h\otimes e_{i_1}^1 \otimes \cdots \otimes e_{i_K}^K)|_{e^j_{i_j}=e_0^j}, \quad j>0, \\
& H_j (h\otimes e_{i_1}^1 \otimes \cdots \otimes e_{i_K}^K)|_{e_{i_j}^j=e_0^j}
 = \frac{i}{\sqrt t} C_j (h\otimes e_{i_1}^1 \otimes \cdots \otimes e_{i_K}^K)|_{e_{i_j}^j=e_1^j}, \quad j>0.
 \end{aligned}
 \end{equation}

At a starting time of an interaction the devices are supposed to be set to their vacuum states,
so that a state $\rho$ on $\HC_0=\C^n$ lifts to $\HC$ as
\[
\rho_{\HC}=\rho \otimes \Om_1 \otimes \cdots \otimes \Om_K.
\]

The observation procedure can be specified by choosing two orthogonal projectors
 $P_0^j$ and $P_1^j$ in the space $\C^2$ of each device (that is in each channel of observation)
 arising from some observables with the spectral decompositions $\sum_l \la_l P_l^j$.
This choice  yields the totality of $2^K$ orthogonal projectors in $\HC$,
\[
I\otimes P_{i_1}^1 \otimes \cdots \otimes  P_{i_K}^K,
\]
so that the possible new non-normalized states after each step of interaction and measurement are
\begin{equation}
\label{newstatedouble}
\tilde \rho_t^{i_1 \cdots i_K}
={\tr}_{p1\cdots K} [(I\otimes P_{i_1}^1 \otimes \cdots \otimes  P_{i_K}^K)
e^{-itH} \rho_{\HC}  e^{itH} (I\otimes P_{i_1}^1 \otimes \cdots \otimes  P_{i_K}^K)],
\end{equation}
where
\begin{equation}
\label{eqMarkchain1m}
\ga_t= e^{-itH} \rho_{\HC}  e^{itH} = e^{-itH}(\rho \otimes \Om_1 \otimes \cdots \otimes \Om_K)  e^{itH},
\end{equation}
and ${\tr}_{p1\cdots K}$ is the partial trace with respect to all spaces, but for $\HC_0$.
These states may occur with the probabilities
\begin{equation}
\label{eqMarkchain2m}
p_{i_1 \cdots i_K}(t)={\tr} \, [\ga_t (I\otimes P_{i_1} \otimes \cdots \otimes P_{i_K})]
={\tr} \tilde \rho_t^{i_1 \cdots i_K}.
\end{equation}

Therefore the multichannel extension of the discrete time {\it Markov chain of successive indirect observations} given by \eqref{eqMarkchain} and \eqref{eqMarkchain1} is given by $2^K$ possible transitions of $\rho$ to the states
\begin{equation}
\label{eqMarkchainmult}
\rho_t^{i_1 \cdots i_K}= \frac{1}{p_{i_1 \cdots i_K}} {\tr}_{p1\cdots K}
[(I\otimes P_{i_1} \otimes \cdots \otimes P_{i_K})
\ga_t (I\otimes P_{i_1} \otimes \cdots \otimes P_{i_K})],
\end{equation}
where $\ga_t$ and the probabilities $p_{i_1 \cdots i_K}$ are given by
\eqref{eqMarkchain1m} and \eqref{eqMarkchain2m}. The transition operator of this Markov chain writes down as
\begin{equation}
\label{eqMarkchain2mult}
U_t f(\rho)=\E f(\rho_t)=\sum_{i_1 \cdots i_K} p_{i_1 \cdots i_K}(t) f(\rho_t^{i_1 \cdots i_K}).
\end{equation}

The operators in $\HC$ are best described in terms of blocks. Namely, writing $\HC=\oplus \HC_{i_1 \cdots i_K}$,
with $\HC_{i_1 \cdots i_K}$ generated by $\HC_0\otimes e_{i_1} \otimes \cdots \otimes e_{i_K}$,
we can represent an operator $\LC$ in $\HC$ by $4^K$ operators
$L_{i_1 \cdots i_K}^{j_1 \cdots j_K}$ in $\HC$, so that
\[
\LC ( h^{i_1 \cdots i_K}\otimes e_{i_1} \otimes \cdots \otimes e_{i_K})
=\sum_{j_1 \cdots j_K} L_{i_1 \cdots i_K}^{j_1 \cdots j_K} h^{i_1 \cdots i_K}
\otimes  e_{j_1} \otimes \cdots \otimes e_{j_K}.
\]

The composition and partial trace in this notations are expressed by the following formulas:
\begin{equation}
\label{eqblockcom}
(\LC_1 \LC_2)_{i_1 \cdots i_K}^{j_1 \cdots j_K}
=\sum_{m_1 \cdots m_K} (\LC_1)_{m_1 \cdots m_K}^{j_1 \cdots j_K}
(\LC_2)^{m_1 \cdots m_K}_{j_1 \cdots j_K},
\end{equation}
\begin{equation}
\label{eqblocktr}
 {\tr}_{p1\cdots K} \LC
 =\sum_{j_1 \cdots j_K} L_{j_1 \cdots j_K}^{j_1 \cdots j_K}.
\end{equation}

For simplicity let us perform detailed calculations for $K=2$ (they are quite similar in the general case).
Thus $\HC=\C^n\otimes \C^2 \otimes \C^2$ and $H=H_0+H_1+H_2$. Let us denote the bases of the two devices
$\{e_k\}$ and $\{f_k\}$ respectively. Formulas \eqref{interHam1} rewrite in a simpler way as
 \[
H_0 (h\otimes e_k \otimes f_j)=A h \otimes e_k \otimes f_j,
 \]
\[
H_1 (h\otimes e_1 \otimes f_j)=-iC_1^* h \otimes e_0 \otimes f_j/ \sqrt t,
\quad
 H_1 (h\otimes e_0 \otimes f_j)=iC_1 h \otimes e_1 \otimes f_j \sqrt t,
 \]
\[
H_2 (h\otimes e_j \otimes f_1)=-iC_2^* h \otimes e_j \otimes f_0 /\sqrt t,
\quad
 H_2 (h\otimes e_j \otimes f_0)=iC_2 h \otimes e_j \otimes f_1  /\sqrt t,
 \]

With the chosen vacuum vectors $e_0=(1,0)$ in the first device and $f_0=(1,0)$ in the second device,
a state $\rho$ on $\HC_0=\C^n$ lifts to $\HC$ as
\[
\rho_{\HC}=\rho \otimes |e_0\rangle \langle e_0| \otimes |f_0\rangle \langle f_0|.
\]

The operators $\LC$ in $\HC$ are described by 16 operators
$L_{jk}^{lm}$ in $\HC$.
To shorten the formulas, let us perform calculations without scaling  $C_j$ (without the factor $1/\sqrt t$)
and will restore the scaling at the end.
In term of the blocks we can write:
\[
(\rho_{\HC})^{ml}_{jk}= \de^m_0  \de^l_0  \de^0_j  \de^0_k \rho.
\]
\[
(H_1)^{ml}_{jk}=i\de^l_k \de^m_{\bar j}(C_1 \de^0_j- C_1^* \de^1_j),
\quad
(H_2)^{ml}_{jk}=i\de^m_j \de^l_{\bar k}(C_2 \de^0_k- C_2^* \de^1_k),
\]
where we have introduced the following notations: for $i$ being $0$ or $1$ we denote
$\bar i$ as being $1$ and $0$ respectively.

By \eqref{eqblockcom} it follows that
\[
[H_1, \rho_{\HC}]^{ml}_{jk}
=i\sum \de^l_q \de^m_{\bar p}(C_1 \de^0_p- C_1^* \de^1_p)\, \de^p_0  \de^q_0  \de^0_j  \de^0_k \rho
-i\sum \de^m_0  \de^l_0  \de^0_p  \de^0_q \rho \, \de^q_k \de^p_{\bar j}(C_1 \de^0_j- C_1^* \de^1_j)
\]
\[
=i \de^0_k \de^l_0 (\de^0_j \de^m_1 C_1\rho+\de^1_j \de^m_0 \rho C_1^*)
=i \de^0_k \de^l_0 \de^m_{\bar j} (\de^0_j C_1\rho+\de^1_j \rho C_1^*).
 \]

Next
\[
(H_1^2)^{ml}_{jk}=\sum (H_1)^{ml}_{pq} (H_1)^{pq}_{jk}
=-\sum \de^l_q \de^m_{\bar p}(C_1 \de^0_p- C_1^* \de^1_p)
\de^q_k \de^p_{\bar j}(C_1 \de^0_j- C_1^* \de^1_j)
\]
\[
=-\de^l_k \de^m_j (C_1\de^1_j-C_1^*\de_j^0)(C_1\de^0_j-C_1^* \de^1_j)
=\de^l_k \de^m_j ( \de^1_j C_1C_1^*+ \de^0_jC_1^* C_1),
\]
\[
(H_2^2)^{ml}_{jk}=\sum (H_2)^{ml}_{pq} (H_2)^{pq}_{jk}
=-\de^m_p \de^l_{\bar q}(C_2 \de^0_q- C_2^* \de^1_q)
\de^p_j \de^q_{\bar k}(C_2 \de^0_k- C_2^* \de^1_k)
\]
\[
=-\de^m_j \de^l_k (C_2 \de^1_k- C_2^* \de_k^0)(C_2 \de^0_k- C_2^* \de^1_k)
=\de^m_j \de^l_k (\de^1_k C_2 C_2^*+ \de^0_kC_2^* C_2),
\]
\[
(H_1 H_2)^{ml}_{jk}=\sum (H_1)^{ml}_{pq} (H_2)^{pq}_{jk}
=-\sum \de^l_q \de^m_{\bar p}(C_1 \de^0_p- C_1^* \de^1_p)
\de^p_j \de^q_{\bar k}(C_2 \de^0_k- C_2^* \de^1_k)
\]
\[
=-\de^l_{\bar k} \de^m_{\bar j}(C_1 \de^0_j- C_1^* \de^1_j)(C_2 \de^0_k- C_2^* \de^1_k),
\]
\[
(H_2 H_1)^{ml}_{jk}=\sum (H_2)^{ml}_{pq} (H_1)^{pq}_{jk}
=-\de^m_p \de^l_{\bar q}(C_2 \de^0_q- C_2^* \de^1_q)
\de^q_k \de^p_{\bar j}(C_1 \de^0_j- C_1^* \de^1_j)
\]
\[
=-\de^l_{\bar k} \de^m_{\bar j}(C_2 \de^0_k- C_2^* \de^1_k)(C_1 \de^0_j- C_1^* \de^1_j),
\]
and
\[
(H_1\rho_{\HC}H_1)^{ml}_{jk}
=(H_1\rho_{\HC})^{ml}_{pq} (H_1)^{pq}_{jk}
=-\de^0_q \de^l_0 \de^m_{\bar p} \de^0_p C_1\rho \,
\de^q_k \de^p_{\bar j}(C_1 \de^0_j- C_1^* \de^1_j)
\]
\[
=\de^0_k \de^l_0 \de^1_j \de^m_1 C_1\rho C_1^*,
\]
\[
(H_2\rho_{\HC}H_2)^{ml}_{jk}
=(H_2\rho_{\HC})^{ml}_{pq} (H_2)^{pq}_{jk}
=-\de^m_0\de^0_p   \de^l_{\bar q}\de^0_q C_2\rho \,
\de^p_j \de^q_{\bar k}(C_2 \de^0_k- C_2^* \de^1_k)
\]
\[
=\de^1_k \de^l_1 \de^0_j \de^m_0 C_2\rho C_2^*,
\]
\[
(H_1\rho_{\HC}H_2)^{ml}_{jk}
=(H_1\rho_{\HC})^{ml}_{pq} (H_2)^{pq}_{jk}
=-\de^0_q \de^l_0 \de^m_{\bar p} \de^0_p C_1\rho \,
\de^p_j \de^q_{\bar k}(C_2 \de^0_k- C_2^* \de^1_k)
\]
\[
= \de^1_0 \de^m_1 \de^0_j \de^1_k C_1\rho C_2^*,
\]
\[
(H_2\rho_{\HC}H_1)^{ml}_{jk}
=(H_2\rho_{\HC})^{ml}_{pq} (H_1)^{pq}_{jk}
=-\de^m_0\de^0_p   \de^l_{\bar q}\de^0_q C_2\rho \,
\de^q_k \de^p_{\bar j}(C_1 \de^0_j- C_1^* \de^1_j)
\]
\[
= \de^1_1 \de^m_0 \de^1_j \de^0_k C_2\rho C_1^*.
\]
Therefore
\[
(H_1+H_2)\rho_{\HC}(H_1+H_2)^{ml}_{jk}=\de^0_k \de^l_0 \de^1_j \de^m_1 C_1\rho C_1^*
+\de^1_k \de^l_1 \de^0_j \de^m_0 C_2\rho C_2^*
\]
\[
+\de^1_0 \de^m_1 \de^0_j \de^1_k C_1\rho C_2^*
+ \de^1_1 \de^m_0 \de^1_j \de^0_k C_2\rho C_1^*.
\]

Next,
\[
\{H_1^2, \rho_{\HC}\}^{ml}_{jk}=(H_1^2)^{ml}_{pq}(\rho_{\HC})^{pq}_{jk}
+(\rho_{\HC})^{ml}_{pq}(H_1^2)^{pq}_{jk}
\]
\[
=\de^l_q \de^m_p (\de^1_p C_1C_1^* +\de^0_p C_1^*C_1) \, \de^p_0  \de^q_0  \de^0_j  \de^0_k \rho
+\de^m_0  \de^l_0  \de^0_p  \de^0_q \rho \, \de^q_k \de^p_j (\de^1_j C_1C_1^* +\de^0_j C_1^*C_1)
=\de^m_0  \de^l_0  \de^0_j  \de^0_k\{C_1^* C_1, \rho\},
\]
\[
\{H_2^2, \rho_{\HC}\}^{ml}_{jk}=(H_2^2)^{ml}_{pq}(\rho_{\HC})^{pq}_{jk}
+(\rho_{\HC})^{ml}_{pq}(H_2^2)^{pq}_{jk}
\]
\[
=\de^m_p \de^l_q (\de^1_q C_2 C_2^*+ \de^0_qC_2^* C_2)\, \de^p_0  \de^q_0  \de^0_j  \de^0_k \rho
+\de^m_0  \de^l_0  \de^0_p  \de^0_q \rho  \, \de^p_j \de^q_k (\de^1_k C_2 C_2^*+ \de^0_kC_2^* C_2)
=\de^m_0  \de^l_0  \de^0_j  \de^0_k\{C_2^* C_2, \rho\},
\]
and
\[
\{H_1 H_2, \rho_{\HC}\}^{ml}_{jk}=(H_1H_2)^{ml}_{pq}(\rho_{\HC})^{pq}_{jk}
+(\rho_{\HC})^{ml}_{pq}(H_1H_2)^{pq}_{jk}
\]
\[
=-\de^l_{\bar q} \de^m_{\bar p}(C_1 \de^0_p- C_1^* \de^1_p)(C_2 \de^0_q- C_2^* \de^1_q)
\de^p_0  \de^q_0  \de^0_j  \de^0_k \rho
-\de^m_0  \de^l_0  \de^0_p  \de^0_q \rho
\de^q_{\bar k} \de^p_{\bar j}(C_1 \de^0_j- C_1^* \de^1_j)(C_2 \de^0_k- C_2^* \de^1_k)
\]
\[
= -\de^l_1  \de^m_1  \de^0_j  \de^0_k C_1C_2\rho - \de^l_0  \de^m_0  \de^1_j  \de^1_k \rho C_1^*C_2^*,
\]

\[
\{H_2 H_1, \rho_{\HC}\}^{ml}_{jk}=(H_2H_1)^{ml}_{pq}(\rho_{\HC})^{pq}_{jk}
+(\rho_{\HC})^{ml}_{pq}(H_2H_1)^{pq}_{jk}
\]
\[
-\de^l_{\bar q} \de^m_{\bar p}(C_2 \de^0_p- C_2^* \de^1_q)(C_1 \de^0_p- C_1^* \de^1_q)
\de^p_0  \de^q_0  \de^0_j  \de^0_k \rho
-\de^m_0  \de^l_0  \de^0_p  \de^0_q \rho
\de^q_{\bar k} \de^p_{\bar j}(C_2 \de^0_k- C_2^* \de^1_k)(C_1 \de^0_j- C_1^* \de^1_j)
\]
\[
=- \de^l_1  \de^m_1  \de^0_j  \de^0_k C_2C_1 \rho -\de^l_0  \de^m_0  \de^1_j  \de^1_k\rho C_2^*C_1^*
\]
Thus,
\[
\{(H_1+H_2)^2,\rho_{\HC}\}^{ml}_{jk}
=\{H_1^2+H_2^2+H_1H_2+H_2H_1, \rho_{\HC}\}^{ml}_{jk}
\]
\[
=\de^m_0  \de^l_0  \de^0_j  \de^0_k\{C_1^*C_1 +C_2^*C_2, \rho\}
-\de^l_1  \de^m_1  \de^0_j  \de^0_k \{C_1,C_2\}\rho
-\de^l_0  \de^m_0  \de^1_j  \de^1_k\rho \{C_1^*, C_2^*\}.
\]

Thus all parts of \eqref{smalltimegroup} are collected.

Let us turn to \eqref{newstatedouble}. From the calculations with a single channel we know that one has to
distinguish diagonal and non-diagonal projectors $P^j_k$. Let us start with the case, when in both devises the projectors
are diagonal, that is
\[
P_0^1=P_0^2=\begin{pmatrix} 1 & 0 \\  0 & 0 \end{pmatrix},
\quad
P_1^1=P_1^2=\begin{pmatrix} 0 & 0 \\  0 & 1 \end{pmatrix}.
\]
Let us calculate
\[
 (I\otimes P_j^1 \otimes P_k^2) \LC (I\otimes P_j^1 \otimes P_k^2)
\]
for arbitrary $\LC$.

 We have
 \[
 (I\otimes P_i^1 \otimes P_r^2) \sum h^{jk} \otimes e_j\otimes f_k=h^{ir},
 \]
 \[
 (I\otimes P_i^1 \otimes P_r^2)^{ml}_{jk}=\de^i_j  \de^r_k  \de^m_i  \de^1_r.
 \]
 So
\[
 ((I\otimes P_i^1 \otimes P_r^2) \LC)^{ml}_{jk}
 =(I\otimes P_i^1 \otimes P_r^2)^{ml}_{pq}  \LC^{pq}_{jk}
=\de^i_p  \de^r_q  \de^m_i  \de^1_r  \LC^{pq}_{jk}
=  \de^m_i  \de^1_r \LC^{ir}_{jk}
\]
and
\[
((I\otimes P_i^1 \otimes P_r^2) \LC (I\otimes P_i^1 \otimes P_r^2))^{ml}_{jk}
=((I\otimes P_i^1 \otimes P_r^2) \LC)^{ml}_{pq}
(I\otimes P_i^1 \otimes P_r^2)^{pq}_{jk}
\]
\[
= \de^m_i  \de^1_r \LC^{ir}_{pq} \de^i_j  \de^r_k  \de^p_i  \de^q_r
=\de^m_i  \de^1_r \de^i_j  \de^r_k \LC^{ir}_{ir}.
\]
Thus
\[
{\tr}_{p12} ((I\otimes P_i^1 \otimes P_r^2) \LC (I\otimes P_i^1 \otimes P_r^2))
=\LC^{ir}_{ir},
\]
and
\[
\tilde \rho_{ir}=(e^{-itH} \rho  e^{itH})^{ir}_{ir},
\quad
p_{ir} ={\tr} (e^{-itH} \rho  e^{itH})^{ir}_{ir}.
\]

Thus we have
\[
[H_1+H_2, \rho_{\HC}]^{jk}_{jk}=0,
\]
\[
(H_1+H_2)\rho_{\HC}(H_1+H_2)^{jk}_{jk}=\de^0_k \de^1_j  C_1\rho C_1^*
+\de^1_k \de^0_j C_2\rho C_2^*,
\]
\[
\{H_1^2+H_2^2+H_1H_2+H_2H_1, \rho_{\HC}\}^{jk}_{jk}
= \de^0_j  \de^0_k\{C_1^*C_1 +C_2^*C_2, \rho\}.
\]

Restoring scaling $C \to C/\sqrt t$ yields approximately
\[
 (e^{-itH} \rho_{\HC}  e^{itH})^{jk}_{jk}
 =(\rho_{\HC}-it [H,\rho_{\HC}]+t^2 (H\rho_{\HC} H-\frac12 \{H^2,\rho_{\HC}\}))^{jk}_{jk}
 \]
 \[
 = \de^0_j  \de^0_k (\rho-it[A,\rho])+t [\de^0_k \de^1_j  C_1\rho C_1^*
+\de^1_k \de^0_j C_2\rho C_2^*-\frac12 \de^0_j  \de^0_k\{C_1^*C_1 +C_2^*C_2, \rho\}]
\]
and thus
\[
\tilde \rho_{jk}= \de^0_j  \de^0_k (\rho-it[A,\rho])+t [\de^0_k \de^1_j  C_1\rho C_1^*
+\de^1_k \de^0_j C_2\rho C_2^*-\frac12 \de^0_j  \de^0_k\{C_1^*C_1 +C_2^*C_2, \rho\}],
\]
\[
p_{jk} =\de^0_j  \de^0_k +t [\de^0_k \de^1_j  {\tr} (C_1\rho C_1^*)
+\de^1_k \de^0_j {\tr} (C_2\rho C_2^*)-\de^0_j  \de^0_k {\tr}((C_1^*C_1 +C_2^*C_2) \rho)].
\]

Thus $p_{11}=0$,
\[
\rho_{00}=\frac{\tilde \rho_{00}}{p_{00}}
=(\rho-it[A,\rho] -\frac12 t\{C_1^*C_1 +C_2^*C_2, \rho\})(1+ t \,{\tr}((C_1^*C_1 +C_2^*C_2) \rho)),
\]
\[
=\rho -it[A,\rho]-\frac12 t\{C_1^*C_1 +C_2^*C_2, \rho\}+t \, {\tr}((C_1^*C_1 +C_2^*C_2) \rho) \rho,
\]
\[
\rho_{10}=\frac{\tilde \rho_{10}}{p_{10}}=\frac {C_1\rho C_1^*} {{\tr} (C_1\rho C_1^*)},
\quad
\rho_{01}=\frac{\tilde \rho_{01}}{p_{01}}=\frac {C_2\rho C_2^*} {{\tr} (C_2\rho C_2^*)}.
\]
Thus we get, up to terms of order $h$ in small $h$, that
\[
\frac{U_h-1}{h} f(\rho)=\frac{1}{h}\sum_{jk} p_{jk} \left[f(\rho_{jk})-f(\rho)\right]
\]
\[
=\frac{1}{h} p_{00}[f(\rho-it[A,\rho]-\frac12 t\{C_1^*C_1 +C_2^*C_2, \rho\}+h \, {\tr}((C_1^*C_1 +C_2^*C_2) \rho) \rho)-f(\rho)]
\]
\[
+\frac{1}{h} p_{10} \left[ f\left(\frac {C_1\rho C_1^*} {{\tr} (C_1\rho C_1^*)}\right) -f(\rho)\right]
+\frac{1}{h} p_{01} \left[ f\left(\frac {C_2\rho C_2^*} {{\tr} (C_2\rho C_2^*)}\right) -f(\rho)\right]
\]
\[
= \left( f'(\rho), -\frac12 \{C_1^*C_1, \rho\}+ {\tr}(C_1 \rho C_1^*) \rho
-\frac12 \{C_2^*C_2, \rho\}+ {\tr}(C_2 \rho C_2^*) \rho\right)
\]
\[
+ {\tr} (C_1\rho C_1^*)\left[ f\left(\frac {C_1\rho C_1^*} {{\tr} (C_1\rho C_1^*)}\right) -f(\rho)\right]
+{\tr} (C_2\rho C_2^*) \left[ f\left(\frac {C_2\rho C_2^*} {{\tr} (C_2\rho C_2^*)}\right) -f(\rho)\right].
\]

Summarising and extending to arbitrary number of channels $k$ we can conclude that we proved the following
extension of Lemma \ref{lemmaongencount}.

 \begin{lemma}
\label{lemmaongencountmultich}
Under the setting considered,
\begin{equation}
\label{eqjumpgener1m}
\|\frac{U_h-1}{h} f -Lf\|\le \sqrt h \ka \|f\|_{C^2(S(\HC_0))}
\end{equation}
for $f\in C^2(S(\HC_0))$, with $L$ given by
\[
L_{count}f(\rho)=-i[A, \rho] \, dt+\sum_{j=1}^K  \left( f'(\rho), -\frac12 \{C_j^*C_j, \rho\}+ {\tr}(C_j \rho C_j^*) \rho\right)
\]
\begin{equation}
\label{eqjumpgener2m}
+ \sum_{j=1}^K \,{\tr}\, (C_j\rho C_j^*)\left[ f\left(\frac {C_j\rho C_j^*} {{\tr} (C_j\rho C_j^*)}\right) -f(\rho)\right].
\end{equation}
\end{lemma}

As a consequence we get the following direct extension of Theorem \ref{thBeleqcount}.

\begin{theorem}
\label{thBeleqcountm}
Let $\HC_0=\C^n$ and $A,C_1, \cdots, C_K$ be operators in $\HC_0$ with $A$ being Hermitian.
Let the projectors defining the measurements be chosen to be diagonal in each channel:
\begin{equation}
\label{eqthBeleqcountm}
P_0^j=\begin{pmatrix} 1 & 0 \\  0 & 0 \end{pmatrix},
\quad
P_1^j=\begin{pmatrix} 0 & 0 \\  0 & 1 \end{pmatrix}
\end{equation}
for all $j=1, \cdots, K$.

Then all statements of Theorem \ref{thBeleqcount} hold for the operator
\eqref{eqjumpgener2m} and Markov semigroups described by the transition operator \eqref{eqMarkchain2mult}.
In particular, estimates \eqref{eq1thBeleqcount} and \eqref{eq2thBeleqcount} hold.
\end{theorem}

\begin{remark} As explained in Remark \ref{remoninfindim} this result extends automatically to the case
of arbitrary separable Hilbert space $\HC$ and bounded operators  $A,C_1, \cdots, C_K$ in it.
\end{remark}

As in the case of a single channel, the process generated by \eqref{eqjumpgener2m}
can be described by the solutions to the SDE of jump type, which takes now the form
\begin{equation}
\label{eqBeleqcountm}
d\rho=- i[A, \rho]\, dt +\sum_j (-\frac12 \{C^*_jC_j,\rho\}+ {\tr} (C_j\rho C^*_j) \rho ) \, dt
+\sum_j\left(\frac{C_j\rho C^*_j}{{\tr} (C_j\rho C^*_j)}-\rho\right) dN^j_t,
\end{equation}
with the counting processes $N_t^j$ are independent and have the position dependent intensities ${\tr} (C^*_jC_j\rho)$.
Equation \eqref{eqBeleqcountm} is the {\it Belavkin quantum filtering SDE} corresponding
to the {\it counting type observation via several channels}.

As suggested by Theorem \ref{thBeleqdif}, exploiting non diagonal pairs of projectors $P_0^j, P_1^j$
should lead to the limiting generator of diffusive type. In fact, performing similar calculations
(which we omit) one arrives at the following general result.

\begin{theorem}
\label{thBeleqmixm}
Let $\HC_0=\C^n$ and $A,C_1, \cdots, C_K$ be operators in $\HC_0$ with $A$ being Hermitian.
Let the projectors defining the measurements are chosen to be diagonal, that is of type
\eqref{eqthBeleqcountm}, for a subset $I\subset \{1, \cdots,K\}$
of the set of channels. And for $j\notin I$ these channels are chosen as non-diagonal,
that is of the form
\begin{equation}
\label{eq1thBeleqmixtm}
P_0^j=\begin{pmatrix} \cos^2 \phi_j & \sin\phi_j \cos \phi_j \\  \sin\phi_j \cos \phi_j  & \sin^2\phi \end{pmatrix},
\quad
P_1^j=\begin{pmatrix} \sin^2 \phi_j & -\sin\phi_j \cos \phi_j \\  -\sin\phi_j \cos \phi_j & \cos^2\phi_j \end{pmatrix},
\end{equation}
with $\phi_j\neq k\pi/2$, $k\in \N$.
Then the limiting generator for the semigroup with the transition operator  \eqref{eqMarkchain2mult}
gets the expression
\[
L_{mix}f(\rho)=\sum_{j\in I}  \left( f'(\rho), -\frac12 \{C_j^*C_j, \rho\}+ {\tr}(C_j \rho C_j^*) \rho\right)
+ \sum_{j\in I} \,{\tr}\, (C_j\rho C_j^*)\left[ f\left(\frac {C_j\rho C_j^*} {{\tr} (C_j\rho C_j^*)}\right) -f(\rho)\right]
\]
\[
+\frac12 \sum_{j\notin I}[(\rho C_j^* + C_j\rho- {\tr} (\rho C^*_j + C_j\rho) \rho)f''(\rho)
(\rho C_j^* + C_j\rho- {\tr} (\rho C^*_j + C_j\rho) \rho)]
\]
\begin{equation}
\label{eq2thBeleqmixtm}
+\sum_{j\notin I} \left(f'(\rho), -\frac12 \{C^*_jC_j,\rho\}+ C_j\rho C_j^* \right)
-(f'(\rho), i[A, \rho]).
\end{equation}
This operator generates
a Feller process $O_t^{\rho}$ in $S(\HC_0)$ and
the corresponding Feller semigroup $T_t$ in $C(S(\HC_0))$ such that claims
 (ii) and (iii) of Theorem \ref{thBeleqdif} hold.
 The Markov process generated by \eqref{eq2thBeleqmixtm} can be given by
 the solutions of the following SDEs in $S(\HC_0)$:
 \[
 d\rho=-i[A, \rho] \, dt
+ \sum_{j\in I} (-\frac12 \{C^*_jC_j,\rho\}+ {\tr} (C_j\rho C^*_j) \rho ) \, dt
+\sum_{j\in I}\left(\frac{C_j\rho C^*_j}{{\tr} (C_j\rho C^*_j)}-\rho\right) dN^j_t
\]
\begin{equation}
\label{eqBeleqmmix}
+\sum_{j\notin I}(-\frac12 \{C_j^*C_j,\rho\}+ C_j\rho C_j^* ) \, dt
+\sum_{j\notin I}\left(\rho C_j^* + C_j\rho-{\tr}\, (\rho C_j^* + C_j\rho)\rho\right) \, dW^j_t,
\end{equation}
where $W_j$ are independent Wiener processes and $N_t^i$ independent jump process of intensity
${\tr} (C_j\rho C^*_j)$.
\end{theorem}

\begin{proof} In the pure diffusive case, that is with empty $I$, the proof is exactly the
same as in Theorem \ref{thBeleqdif}. For the general case one only has to show
that operator $L_{mix}$ generates a Feller process in $S(\HC_0)$ preserving the sets of smooth functions
(other arguments are again the same).
Two proofs for proving this fact can be suggested.
(i) One  starts with generator $\tilde L_{mix}$ obtained from \eqref{eq2thBeleqmixtm}
by ignoring the jump part. This is a well-defined diffusion operator and by the same methods as
in  Theorem \ref{thBeleqdif} one shows that it generates a Feller processes in $S(\HC_0)$.
But the jump part of \eqref{eq2thBeleqmixtm} is a bounded operator preserving positivity and smoothness.
Hence it can be dealt with straightforwardly via the perturbation theory.
(ii) Each of the two parts of \eqref{eq2thBeleqmixtm}, related to $I$ and its complement,
generates a well-defined Feller process in $S(\HC_0)$ preserving smoothness (of arbitrary order in fact).
Hence one can derive that the sum of these operators generates  a well-defined Feller process in $S(\HC_0)$
via the Lie-Trotter formula, namely from Theorem 5.3.1 of \cite{Kolbook11}.
\end{proof}

\begin{remark}
The Markov chain of multichannel measurement that we are using is a bit different from the one used in
\cite{Pellegrini10a}, where measurement is based on a single operator $R$ in the device (no different channels),
and counting and diffusive parts of the generator arise from different projectors linked to different eigenspaces
of this operator. As was already mentioned the method of \cite{Pellegrini10a}
did not provide the rates of convergence.
\end{remark}

When $I$ is empty, $L_{mix}$ turns to $L_{dif}$ describing the multichannel observations of diffusive type.

\section{Fractional quantum stochastic filtering}
\label{secfraceq}

Now everything is ready for our main result: the derivation of the fractional equations of quantum stochastic filtering.
As was shown above the standard Belavkin equations of quantum filtering can be obtained as the scaled limits of
the sequences of discrete observations. The main assumption for each of the approximating processes was that the time
between successive measurement is either constant (discrete Markov chain approximation) or is exponentially distributed
(continuous time Markov chain approximation). Of course there is no a priori reasons for these assumptions.
 And in fact in several domains of physics it turned out to be more appropriate to model times between successive
 events by random variables from the domains of attraction of a stable law, that is via CTRW.

Our next result is a direct consequence of Theorem  \ref{thBeleqmixm} and Proposition \ref{propCTRW}.

 \begin{theorem}
 \label{mainth}
 Under the assumptions of Theorem \ref{thBeleqmixm} let the Markov chain \eqref{eqMarkchainmult}
is modified in such a way that the laws of transitions $\rho \to \rho_t^{i_1 \cdots i_K}$
remain unchanged, by the time between transitions is taken as scaled random variable from the domain
of attraction of a $\be$-stable law, that is as $T_i^h=h^{1/\be} T_i$  from Proposition \ref{propCTRW}.
Then the corresponding generalized CTRW processes \eqref{ctrwposdep}  built from the transition operator
\eqref{eqMarkchain2mult} converge to the process $O^{\rho}_{\si_t}$ obtained from the process
$O^{\rho}_t$  of Theorem \ref{thBeleqmixm} via subordination by the inverse stable process
$\si_t=\max \{y: S_y \le t\}$. Moreover, the functions $f_t(x)=\E (T_{\si_t}f)(x)$
satisfy the fractional Caputo-Djerbashian equation
\eqref{eqfraceq} with the generator $L=L_{mix}$ given by \eqref{eq2thBeleqmixtm}.
 \end{theorem}

 As noted at the end of Appendix C the fractional derivative $D^{\be}_{0+\star}$ is a particular case of a class
 of mixed fractional derivatives \eqref{eqfraceq}. Therefore, under appropriately organised scaled times
 between the acts of measurements the limiting evolution will satisfy a more general fractional equation
 \begin{equation}
\label{eqfraceqquBe}
 D^{(\nu)}_{0+\star}f_t(x)=L_{mix}f_t(x), \quad f_0(x)=f(x),
\end{equation}
with $D^{\nu}$ given by \eqref{eqfraceq2}.

When only one type of observation channels is used, equation \eqref{eqfraceqquBe} simplifies to the case,
when either $L_{count}$ or $L_{dif}$ are places instead of $L_{mix}$.

Equations \eqref{eqfraceqquBe} (and their particular cases with fractional derivative $D^{\be}$ of order $\be$)
represent the fractional analogs of the process of quantum stochastic filtering. These equations can be
also considered as the new equations of fractional quantum mechanics. They are different from the fractional
Schr\"odinger equations suggested in  \cite{Lask02} and extensively studied recently.

Equations \eqref{eqfraceqquBe} describe the process of continuous quantum control and filtering on the level
of the evolution of averages. On the 'micro-level' of SDEs  \eqref{eqBeleqmmix} these equations correspond
to stopping the solutions of these SDEs at a random time $\si_t$ given by the inverse of a L\'evy subordinator.

\section{Fractional quantum control and games}
\label{secfraceqHJB}

The theory of quantum filtering reduces the analysis of quantum dynamic control and games
to the controlled version of evolutions \eqref{eqBeleqmmix}. The simplest situation concerns the case
 when the homodyne device is fixed, that is the operators $C_j$ and the projectors $P_i^j$ are fixed,
 and the players can control the individual Hamiltonian $H_0$ of the atom, say,
 by applying appropriate electric or magnetic fields to the atom. Thus equations
 \eqref{eqBeleqmmix} become modified by allowing $H_0$ to depend on one or several control parameters.
 The so-called separation principle states
 (see \cite{BoutHanQuantumSepar}) that the effective control of an observed quantum system
 (that can be based in principle on the whole history of the interaction of the atom and optical devices)
 can be reduced to the Markovian feedback control of the quantum filtering equation, with the feedback
 at each moment depending only on the current (filtered) state of the atom.

In the present case of CTRW modeling of the process of measurements the problem of control becomes the
problem of control of scaled CTRW. The theory of such control was built in the series of papers
\cite{KolVer14} - \cite{KolVer17}. The main result is that in the scaling limit the cost functions
is a solution of the fractional Hamilton-Jacobi equation. In the present context and in game-theoretic setting
it implies the following.
Let us consider the controlled version of the process  $O^{\rho}_{\si_t}$  from Theorem \ref{mainth},
where the individual Hamiltonian is now $\tilde H_0=H_0+uH_0^1+v H_0^2$ and it depends on control parameters
 $u,v$ of two players from compact sets $U$ and $V$ respectively. Suppose that it is possible
to choose new $u,v$ directly after each act of measurement, and thus a control strategy is the sequence
$(u_1, v_1), (u_2,v_2), \cdots )$ of controls applied after each act of measurement, with each $(u_j,v_j)$
applied after $j$th act of measurement and depending on the history of the process until that time.
The case of a pure control (not a game) corresponds to the choice $V=0$ and is thus automatically included.
Assume that players $I$ and $II$ play a standard dynamic zero-sum game with a finite time horizon $T$
meaning that the objective of $I$ is to maximize the payoff
 \begin{equation}
\label{eqcostfun}
P(t; u(.), v(.)) =\E [\int_t^T   {\tr} \, (J \rho_s) \, ds +{\tr} \, (F \rho_T)],
\end{equation}
where $J$ and $F$ are some operators expressing the current and the terminal costs of the game
(they may depend on $u$ and $v$, but we exclude this case just for simplicity) and $W$ is the collection
of all noises involved in \eqref{eqBeleqmmix} (both diffusive and Poisson).
Then under the scaling limit of Theorem \ref{mainth} the optimal cost function
\begin{equation}
\label{eqcostfunopt}
S_t(\rho) =\max_{u(.)} \min_{v(.)}P(t; u(.), v(.))=\min_{v(.)} \max_{u(.)} P(t; u(.), v(.))
\end{equation}
will satisfy the following {\it fractional HJB-Isaacs equation of the CTRW modeling of quantum games}:
 \begin{equation}
\label{eqHJB}
D^{\nu}_{0+\star}S_t(\rho)=\max_u (f'(\rho), i[\rho, uH_0^1]) +\min_v (f'(\rho), i[\rho, vH_0^2])
+{\tr} \, (J \rho_t)+ L_{mix}S_t(\rho).
\end{equation}

In \cite{KolVer14} this equation was derived heuristically, in the general framework of controlled
CTRW by the dynamic programming approach. As usual in optimal control theory, to justify the derivation
one has to show the well-posedness of the limiting HJB equation
and then to prove the verification theorem, a classical reference is \cite{FlemSon}. For some cases
of CTRWs  this was performed in \cite{KolVer17}.

In the present fractional quantum case this problem will be considered elsewhere.
The additional complexity of this equation is related to the fact that the state space is a
rather nontrivial set of positive matrices with the unit trace. One can reduce the complexity
by looking at the dynamics of pure states only. But the set of pure states is not a Euclidean space,
but a manifold. In the finite-dimensional setting this manifold is the complex projective space $\C P^n$.

Let us mention that in the non-fractional case, that is with the usual derivative $\pa/\pa t$ instead of
$D^{\nu}_{0+\star}$ in \eqref{eqHJB}, the well-posedness of (the analogs of) equation \eqref{eqHJB} was proved
in \cite{Kol92}, for a special model of pumping a laser with a counting measurement, with some particular
solutions calculated explicitly,
and in \cite{KolDynQuGames}, for a special arrangements of diffusive measuring devises that ensured that the
diffusive part of operator $L_{dif}$ was nondegenerate and therefore the optimal control problem was reduced
to the drift control of the diffusions on a Riemannian manifold  $\C P^n$.

\section{Other Markov approximations and unbounded generators}
\label{secunbound}

We commented above on the possible extension to infinite-dimensional Hilbert spaces. However, for all approximations
the assumption of boundedness of all operators involved seemed to be essential in the derivation given,
at least of the coupling operators $C_j$ (unboundedness of $A$ can be possibly treated via the interaction representation).
However, the quantum filtering equations are used also in the standard setting of quantum mechanics. The mostly studied
case is that of the standard Hamiltonian $H=-\De+V(x)$ in $L^2(\R^d)$ and the coupling operators being either position
(multiplication by $x$) or momentum operators. Different Markov chain approximations may be used to derive the
filtering equation in this case.

A powerful approach was suggested by Belavkin in \cite{Be195}: to use the von Neumann model
of unsharp measurement. In this model the effect of measurement for the product state $\phi(x)f(y)$ of an atom and
a measuring device, a pointer, is given by the shift
\[
U: \phi (x) f(y) \mapsto \phi(x) f(y-ax).
\]
Here both $\phi$ and $f$ are from $L^2(\R^d)$, and $f>0$ describes the stationary state of a pointer
(the analog of the vacuum state in our modeling above).
Projecting on the state of an atom this yields the transition
\begin{equation}
\label{ea0}
G(y) \colon \quad \phi(x) \mapsto \phi_y(x)=\phi(x) f(y-ax)/f(y),
\end{equation}
depending on the observed position $y$ of the pointer. Assuming the evolution of the atom
during time $t$ between the moments of measurements to be given by a Hamiltonian $A$,
the transition of a Markov chain of sequential measurements become
\begin{equation}
\label{ea1}
\phi \mapsto \phi_{t,y}(x)=(e^{-iAt}\phi)(x) f(y-ax)/f(y).
\end{equation}
After an appropriate scaling from this Markov chain one derives the diffusive filtering SDE
\eqref{eqqufiBnonlins} with $C=x$ (the multiplication operator), that is directly the filtering equation
for pure states, see detail in Appendix to \cite{BelKol}.
The model can be extended to more general situations, but seems to be linked with a specific von Neumann
instantaneous interaction. For the well-posedness of these kind of diffusive SDEs we can refer to
\cite{Holevo91}, \cite{FagMor} and references therein.

The derivation of the fractional version of this equation, as well as the fractional
control of Section  \ref{secfraceqHJB} can be performed in this setting in the same way as above.

\section{Appendix A. Convergence of semigroups}
\label{secconbergsem}

Here we collect the results on the convergence of Markov semigroups and CTRW, which form the
the theoretical basis for our derivations of the filtering equations.

It is well known that the convergence of the generators on the core of the limiting generator
implies the convergence of semigroups. We shall use a version of this result with the rates,
namely the following result, given in Theorem 8.1.1 of \cite{Kolbook11}.

\begin{prop}
\label{propconvsemigr}

 Let $F_t=e^{tL}$ be a strongly continuous semigroup in a Banach space $B$ with a norm $\|.\|_B$,
 generate by an operator $L$,
having a core $D$, which is itself a Banach space with a norm $\|.\|_D\ge \|.\|_B$ so that $L\in \LC(D,B)$.
Let $F_t$ be also a bounded semigroup in $D$
such that $\|F_t\|_{D\to D} \le C_D(T)$ with a constant $C_D(T)$ uniformly for $t\in [0,T]$.

(i) Let $F_t^h$, $h>0$,  be a family of strongly continuous contraction semigroups in a Banach space $B$
with bounded generators $L_h$ such that
\[
\|L_hf-Lf\|_B \le \ep_h \|f\|_D
\]
for all $f\in D$ and some $\ep_h$ such that $\ep_h\to 0$ as $h\to 0$.
 Then the semigroups $F_t^h$ converge strongly to the semigroup $F_t$, as $h\to 0$, and
\begin{equation}
\label{eq1propconvsemigr}
\|F_t^hf -F_tf\|_B \le t \ep_h C_D(T)\|L\|_{D\to B}.
\end{equation}

(ii) Let $U_h$ be a family of contractions in $B$ such that
\begin{equation}
\label{eq2propconvsemigr}
\|\left(\frac{U_h-1}{h} -L\right)f\|_B \le \ep_h \|f\|_D,
\end{equation}
and
 \begin{equation}
\label{eq3propconvsemigr}
\|\left(\frac{F_h-1}{h} -L\right)f\|_B \le \ka_h \|f\|_D,
\end{equation}
with $\ep_h \to 0$ and $\ka_h\to 0$, as $h\to 0$.
Then the scaled discrete semigroups $(U_h)^{[t/h]}$ converge to the semigroup $F_t$
and moreover
 \begin{equation}
\label{eq4propconvsemigr}
\sup_{s\le t}\|(U_h)^{[s/h]} -F_sf\|_B \le (\ka_h+\ep_h)t \|f\|_B.
\end{equation}
\end{prop}

Additional condition \eqref{eq3propconvsemigr} makes working with discrete approximation a bit more subtle,
than with the continuous chain approximations. Effectively to get \eqref{eq3propconvsemigr}
one needs a deeper regularity. Namely one should have another core $\tilde D$ such that $D\subset \tilde D\subset B$
with $L\in \LC(D,\tilde D) \cap \LC(\tilde D,B)$. In this case it is easy to see that
\begin{equation}
\label{eq5propconvsemigr}
\|\left(\frac{F_h-1}{h} -L\right)f\|_B \le h \|L\|_{D,\tilde D} \|L\|_{\tilde D,B}\|f\|_D.
\end{equation}

\section{Appendix B. Deterministic motions with random jumps}
\label{secdetandjump}

Let us look at the Cauchy problem
 \begin{equation}
\label{eqdetjump}
\frac{\pa f_t}{\pa t} =(\nabla f_t, b(x))+Lf_t(x), \quad f_0(x) \, \text{given},
\end{equation}
with the simplest jump-type operator
\[
L_f(x)=\sum_{j=1}^J f(Y_j(x)-x),
\]
where $x\in \R^d$, $\nabla f=\pa f/\pa x$ and $b,Y_j:\R^d\to \R^d$ are given bounded smooth functions.
It is more or less obvious that the resolving operators of the Cauchy problem \eqref{eqdetjump}
form a semigroup of contractions in the space $C(\R^d)$ preserving the spaces of
smooth functions. Let us make a precise statement. The simplest way to see it is via the 'interaction
representation'. Namely, let $X_t(x)$ denote the solution to the Cauchy problem $\dot X_t(x)=b(X_t(x))$,
$X_0(x)=x$, and let us change the unknown function $f$ in \eqref{eqdetjump} to $\phi$ via the equation
$f(x)=\phi(X_t(x))$. Direct substitution shows that $\phi$ solves the Cauchy problem
 \begin{equation}
\label{eqdetjump1}
\frac{\pa \phi_t}{\pa t} =L_t\phi_t(x)=\sum_{j=1}^J \phi((X_t(Y_j(X_{-t}(x))))-x), \quad \phi_0=f_0.
\end{equation}
Since $L_t$ is a bounded operator, this Cauchy problem can be solved by the convergence series
 over the powers of $L_t$. This leads to the following result.

\begin{prop}
\label{propdetjump}
Let $b, Y_j\in C^2(\R^d)$, $j=1, \cdots, J$. Then the resolving operators $R_t$ of the Cauchy problem
\eqref{eqdetjump} form a semigroups of contractions in $C(\R^d)$ such that the spaces $C^1(\R^d)$ and $C^2(\R^d)$
are invariant and $R_t$ form semigroups of operators in these spaces that are uniformly bounded for $\in [0,T]$ with any $T$.
\end{prop}

We need an extension of this result for the subsets of $\R^d$. The main tool is the following classical theorem of Brezis,
which we formulate in its simplest form referring to proofs, extensions and history to \cite{Redheffer}.

\begin{theorem}
\label{thBrezis}
Let $b(x):K\to \R^d $ be a Lipschitz continuous function, where $K$ is a convex closed subset of $\R^d$, such that
\begin{equation}
\label{eqBrezis}
\lim_{h\to 0_+} \frac{d(y+hb(x),K)}{h}=0
\end{equation}
 for any $x\in K$, where $d(z,K)$ denotes the distance between a point $z$ and the set $K$.
 Then $K$ is flow invariant. More precisely, for any $x\in K$ there exists a unique solution $X_t(x)$
 of the equation $\dot X_t(x)=b(X_t(x))$ with the initial condition $x$ that belongs to $K$ for all $t$.
\end{theorem}

As a direct consequence we get the following extension of Proposition \ref{propdetjump}.

\begin{prop}
\label{propdetjumpdom}
Let $K$ be a convex compact subset of $\R^d$ and  $b:K\to \R^d$, $Y_j:K\to K$ be twice continuously differentiable functions.
Let $b$ satisfy the assumptions of Theorem \ref{thBrezis}. Then  the resolving operators $R_t$ of the Cauchy problem
\eqref{eqdetjump} form a semigroups of contractions in $C(K)$ such that the spaces $C^1(K)$ and $C^2(K)$
are invariant and $R_t$ are uniformly bounded operators in these spaces for $\in [0,T]$ with any $T$.
\end{prop}

\section{Appendix C. Position dependent CTRW}
\label{secCTRW}

Here we recall the basic result on the convergence of continuous time random walks (CTRW).

Suppose $T_1^h,T_2^h, \cdots $ is a sequence of i.i.d.
random variables in $\R_+$ such that the distribution of each $T_i^h$
is given by a probability measure $\mu_{time}^h (dt)$
on $\R_+$, that depend on a positive (scaling) parameter $h$. Let
\begin{equation}
\label{eqinversewalk}
N_t^h=\max \{ n: \sum_{i=1}^n T_i^h \le t\}.
\end{equation}

Suppose $X_1^h,X_2^h, \cdots $ is a sequence of i.i.d.
random variables in $\R^d$, such that the distribution of each $X_i^h$
is given by a probability measure $\mu_{space}^h (dt)$, that depends on $h$.
The standard (scaled) {\it continuous time random walk}\index{continuous time random walk (CTRW)} (CTRW)
 is a random process given by the random sum
\[
\sum_{j=1}^{N_t^h} X_i^h.
\]

In position dependent CTRW the jumps $X_i^h$ are not independent, but each $X_i^h$
depends on the position of the process before this jump. The natural general formulation
 can be given in terms of discrete Markov chains as follows.
Let $U_h$ be a transition operator of a discrete time Markov chain $O^h_n(x)$ in $\R^d$
depending on a positive parameter $h$, so that
\begin{equation}
\label{transitoperMar}
U_hf(x)=\E O^h_1(x)=\int f(y) \mu^h(x, dy),
\end{equation}
with some family of stochastic kernels  $\mu^h(x, dy)$ such that $U_h$
is a bounded operator either in the space $C(K)$ with a compact convex subset $K$ of $\R^d$
or in the space $C_{\infty}(\R^d)$ of continuous functions vanishing at infinity.
For our purposes we need only
the operators of the type
\[
U_hf(x)=\E O^h_1(x)=\sum_{j=1}^J  f(Y_j^h(x)) p_j(x)^h,
\]
with a family of continuous mappings $Y_j^h:\R^d\to \R^d$ and the probability laws $\{p_1^h, \cdots, p_J^h\}$.

Suppose $T_1^h,T_2^h, \cdots $ is a sequence of random variables introduced above, and independent of $O^h_n(x)$.
The process
\begin{equation}
\label{ctrwposdep}
 O^h_{N_t^h}(x)
\end{equation}
is a generalized scaled (position dependent)
 {\it continuous time random walk}\index{continuous time random walk (CTRW)} (CTRW) arising from
$U_h$ and $\mu^h_{time}$.

The CTRW were introduced in \cite{MW}. They found numerous applications in physics.
The scaling limits of these CTRW were analysed by many authors, see e.g. \cite{KKU}, \cite{MS},
\cite{MS2}.  The scaling limit for the position dependent CTRW was developed in \cite{Kol7}.
Formally in  \cite{Kol7} it was developed not in full generality, but for the case of the spacial
process $O^h_n(x)$ converging to a stable process. However, the arguments of  \cite{Kol7} were completely
general and did not depend on this assumption. The only point used was that $O^h_n(x)$ converge in
the sense of  Proposition \ref{propconvsemigr} (ii). For completeness let us formulate the result
\cite{Kol7} in a slightly modified version that we need in this paper and present a short proof
with essentially simplified arguments from  \cite{Kol7} (see also Chapter 8 in \cite{Kolbook11}).

As an auxiliary result we need the standard functional limit theorem for the random-walk-approximation
of stable laws, see e.g. \cite{GnedKor} and \cite{MS2} and references therein for various proofs.
\begin{prop}
\label{limtheorstable}
Let a positive random variables $T$ belong to the domain of attraction of a $\be$-stable law,
$\be\in (0,1)$, in the sense that
\begin{equation}
\label{eq1propCTRW}
 \P (T>m)\sim \frac{1}{\be m^{\be}}
\end{equation}
(the sign $\sim$ means here that the ratio tends to $1$, as $m\to \infty$).
Let $T_i$ be a sequence of i.i.d. random variables from
the domain of attraction of a $\be$-stable law and let
\[
\Phi_t^h=\sum_{i=1}^{[t/h]} h^{1/\al} T_i
\]
be a scaled random walk based on $T_i$, $h>0$, and $S_t$ a $\be$-stable L\'evy subordinator, that is
a L\'evy process in $\R_+$ generated by the stable generator
\[
L_{\be}(x)=\int \frac{f(x+y)-f(x)}{y^{1+\be}} dy
\]
(which up to a multiplier  represents the fractional derivative $d^{\be}/d(-x)^{\be}$).
Then  $\Phi_t^h \to S_t$ in distribution, as $h\to 0$.
\end{prop}

The next result is from \cite{Kol7}, though modified and simplified.
\begin{prop}
\label{propCTRW}
Let the random variables $T_i^h=h^{1/\be} T_i$, where i.i.d. random variables $T_i$ belong to the
domain of attraction of a $\be$-stable law, $S_t$ be a $\be$-stable L\'evy suboridinator
and
\[
\si_y=\max \{t: S_t \le y\}
\]
be its inverse process.
Let a family of contractions \eqref{transitoperMar}
satisfy \eqref{eq2propconvsemigr} with an operator $L$ generating a Feller process $F_t$.
Then
\[
\E U_h^s|_{s=[N_t^h/h]} \to \E F_{\si_t}, \quad h\to 0,
\]
strongly as contraction operators in $C(K)$ or $C_{\infty}(\R^d)$.
\end{prop}

\begin{remark} This Proposition directly implies the following statement about the processes:
the subordinated Markov chains \eqref{ctrwposdep}, that is the scaled CTRW, converge in distribution
to the process generated by $L$ and subordinated by the inverse of the L\'evy $\be$-subordinator.
\end{remark}

\begin{proof}
By the density arguments it is sufficient to show that
\[
\|\E U_h^{[s/h]}|_{s=N_t^h}f - \E F_{\si_t}f\| \to 0
\]
for functions $f$ from the domain of $L$.
We have
\[
\|\E U_h^{[s/h]}|_{s=N_t^h}f - \E F_{\si_t}f\| \le I+II,
\]
with
\[
I=\|\E U_h^{[s/h]}|_{s=N_t^h}f - \E F_{N^h_t}f\|, \quad II=\|\E F_{N^h_t}f - \E F_{\si_t}f\|.
\]
To estimate I we write
\[
I=\int_0^{\infty}(U_h^{[s/h]}f-F_sf)\mu^h_t(ds)=\int_0^K(U_h^{[s/h]}f-F_sf)\mu^h_y(ds)
+\int_K^{\infty}(U_h^{[s/h]}f-F_sf)\mu^h_t(ds),
\]
where $\mu^h_t$ is the distribution of $N_t^h$. Choosing $K$ large
enough we can make the second integral arbitrary small uniformly in $h$. And then by \eqref{eq4propconvsemigr}
we can make the first integral arbitrary small by choosing small enough $h$ (and uniformly in $t$ from
compact sets).
It remains II. Integrating by parts we get the following:
\[
II=\| \E e^{N_t^h L}f-\E e^{\si_t L} f\|
\]
\[
=\|\int_0^{\infty} \frac{\pa}{\pa s} (e^{sL}f) (\P(\si_t\le s) -\P(N_t^h\le s)) \, ds\|
\]
\[
= \|\int_0^{\infty}  L e^{sL}f (\P(S_s> t) -\P(\Phi_s^h>t)) \, ds\|.
\]
By \eqref{limtheorstable}, $\P(S_s> t) \to \P(\Phi_s^h>t)$ as $h\to 0$. Therefore $II \to 0$ by the dominated
convergence, as $h\to 0$.
\end{proof}

\begin{remark}
From this proof it is seen how to get some explicit rates of convergence. We are not going to give details.
\end{remark}

It is well known, see e.g. \cite{SZ} and detailed presentations in monographs \cite{Meerbook}, \cite{Kolbook19},
that the subordinated limiting evolution described by the operators
$\E F_{\si_t}$ solves fractional in time differential equations. Namely,
under the conditions of Proposition \ref{propCTRW}, the function $f_t(x)=\E (F_{\si_t}f)(x)$
satisfies the equation
\begin{equation}
\label{eqfraceq}
 D^{\be}_{0+\star}f_t(x)=Lf(x), \quad f_0(x)=f(x),
\end{equation}
where  $D^{\be}_{0+\star}$ is the Caputo-Djerbashian derivative of order $\be$ acting on the variable $t$, and the
operator $L$ acts on the variable $x$.

Recall that  a L\'evy subordinator is a process generated by the operator
\begin{equation}
\label{eqLevysubord}
L_{\nu}f(x)=\int_0^{\infty} f(x+y)-f(x)) \nu (dy),
\end{equation}
where $\nu$ is a one-sided L\'evy measure, that is , it satisfies the condition $\int \min(1,y) \nu(dy)<\infty$.
Proposition \ref{propCTRW} is based on the central limit for stable laws stating the convergence $\Phi_t^h \to S_t$
of random walks approximations to a stable L\'evy subordinator. If scaled random walks $\Phi_t^h$ are designed in such a way
that they approximate an arbitrary L\'evy subordinator, that is,
$\Phi_t^h \to S_t$ with $S_t$ generated by \eqref{eqLevysubord},
then similar arguments show that
\[
\E U_h^s|_{s=[N_t^h/h]} \to \E F_{\si_t}, \quad h\to 0,
\]
where
\[
\si_y=\max \{t: S_t \le y\}, \quad
\N_y^h=\max \{t: \Phi_t^h \le y\}.
\]
In this case the functions $f_t(x)=\E (F_{\si_t}f)(x)$
satisfy the equation
\begin{equation}
\label{eqfraceq1}
 D^{(\nu)}_{0+\star}f_t(x)=L_tf(x), \quad f_0(x)=f(x),
\end{equation}
see e.g. \cite{Kol7}, \cite{Kol15},
where  $D^{(\nu)}_{0+\star}$ is the generalised Caputo-type mixed fractional derivative
defined by the equation
\begin{equation}
\label{eqfraceq2}
 D^{(\nu)}_{0+\star}f_t=\int_0^t (f_{t-s}-f_t)\nu(ds)+ (f_0-f_t)\int_t^{\infty} \nu(ds).
\end{equation}
The derivative $D^{\be}_{0+\star}$ in \eqref{eqfraceq} corresponds to $\nu(dy)=y^{-1-\be} dy$.

\end{document}